\newif\ifprocs
\procstrue
\procsfalse
\ifprocs
\documentclass[a4paper,UKenglish]{lipics-v2016}
 
\usepackage{microtype}

\title{Characterizing Demand Graphs for  (Fixed-Parameter) Shallow-Light Steiner Network}
\titlerunning{Characterizing Demand Graphs for Shallow-Light Steiner Network} 

\author[1]{Amy Babay}
\author[2]{Michael Dinitz}
\author[3]{Zeyu Zhang}
\affil[1]{Department of Computer Science, Johns Hopkins University\\
  \texttt{babay@cs.jhu.edu}}
\affil[2]{Department of Computer Science, Johns Hopkins University\\
  \texttt{mdinitz@cs.jhu.edu}}
\affil[3]{Department of Computer Science, Johns Hopkins University\\
  \texttt{zyzhang92@gmail.com}}
\authorrunning{A.~Babay, M.~Dinitz, and Z.~Zhang} 

\Copyright{Amy Babay, Michael Dinitz and Zeyu Zhang}

\subjclass{Dummy classification -- please refer to \url{http://www.acm.org/about/class/ccs98-html}}
\keywords{keywords}

\EventEditors{Ioannis Chatzigiannakis, Christos Kaklamanis, Daniel Marx, and Don Sannella}
\EventNoEds{4}
\EventLongTitle{45th International Colloquium on Automata, Languages, and Programming (ICALP 2018)}
\EventShortTitle{ICALP 2018}
\EventAcronym{ICALP}
\EventYear{2018}
\EventDate{July 9--13, 2018}
\EventLocation{Prague, Czech Republic}
\EventLogo{eatcs}
\SeriesVolume{80}
\ArticleNo{}

\theoremstyle{plain}
\newtheorem{claim}[theorem]{Claim}

\else

\documentclass[letterpaper, 11pt]{article}	
\usepackage[in]{fullpage}
\usepackage{amsthm}
\newtheorem{lemma}{Lemma}[section]
\newtheorem{theorem}[lemma]{Theorem}
\newtheorem{definition}[lemma]{Definition}
\newtheorem{corollary}[lemma]{Corollary}
\newtheorem{claim}[lemma]{Claim}
\title{Characterizing Demand Graphs for the (Fixed-Parameter) Shallow-Light Steiner Network Problem}
\author{Amy Babay \qquad  \qquad Michael Dinitz \qquad \qquad Zeyu Zhang\\Department of Computer Science\\Johns Hopkins University}

\fi

\usepackage{latexsym}
\usepackage{amssymb}
\usepackage{amsmath}
\usepackage{amsfonts}
\usepackage{graphicx}
\usepackage{dsfont}
\usepackage{algorithm}
\usepackage{algpseudocode}
\usepackage{xcolor}

\newcounter{note}[section]

\begin{document}

\maketitle

\abstract
We consider the \textsc{Shallow-Light Steiner Network} problem from a fixed-parameter perspective. Given a graph $G$, a distance bound $L$, and $p$ pairs of vertices $(s_1,t_1),\mathellipsis,(s_p,t_p)$, the objective is to find a minimum-cost subgraph $G'$ such that $s_i$ and $t_i$ have distance at most $L$ in $G'$ (for every $i \in [p]$). Our main result is on the fixed-parameter tractability of this problem with parameter $p$. We exactly characterize the demand structures that make the problem ``easy'', and give \textsc{FPT} algorithms for those cases.  In all other cases, we show that the problem is \textsc{W}$[1]$-hard.  We also extend our results to handle general edge lengths and costs, precisely characterizing which demands allow for good \textsc{FPT} approximation algorithms and which demands remain \textsc{W}$[1]$-hard even to approximate.

\section{Introduction}\label{sec:intro}

In many network design problems we are given a graph $G = (V, E)$ and some demand pairs $(s_1, t_1), (s_2, t_2), \dots, (s_p, t_p) \subseteq V \times V$, and are asked to find the ``best'' (usually minimum-cost) subgraph in which every demand pair satisfies some type of connectivity requirement.  In the simplest case, if the demands are all pairs and the connectivity requirement is just to be connected, then this is the classical \textsc{Minimum Spanning Tree} problem.  If we consider other classes of demands, then we get more difficult but still classical problems.  Most notably, if the demands form a star (or any connected graph on $V$), then we have the famous \textsc{Steiner Tree} problem.  If the demands are completely arbitrary, then we have the \textsc{Steiner Forest} problem.  Both problems are known to be in \textsc{FPT} with parameter $p$~\cite{dreyfus1971steiner} (i.e., they can be solved in $f(p)\cdot poly(n)$ time for some function $f$).


There are many obvious generalizations of \textsc{Steiner Tree} and \textsc{Steiner Forest} of the general network design flavor given above.  We will be particularly concerned with \emph{length-bounded} variants, which are related to (but still quite different from) \emph{directed} variants.  In \textsc{Directed Steiner Tree} (DST) the input graph is directed and the demands are a directed star (either into or out of the root), while in \textsc{Directed Steiner Network} (DSN) the input graph and demands are both directed, but the demands are an arbitrary subset of $V \times V$.  Both have been well-studied (e.g., \cite{charikar1999approximation,zosin2002directed,chekuri2011set,chlamtac2017approximating,abboud2018reachability}), and in particular it is known that the same basic dynamic programming algorithm used for \textsc{Steiner Tree} will also give an \textsc{FPT} algorithm for \textsc{DST}.  However, \textsc{DSN} is known to be \textsc{W}$[1]$-hard, so it is not believed to be in \textsc{FPT}~\cite{feldmann2017complexity}.

In the length-bounded setting, we typically assume that the input graph and demands are undirected but each demand has a distance bound, and a solution is only feasible if every demand is connected within distance at most the given bound (rather than just being connected).  One of the most basic problems of this form is the \textsc{Shallow-Light Steiner Tree} problem (SLST), where the demands form a star with root $r = s_1 = s_2 = \dots = s_p$ and there is a global length bound $L$ (so in any feasible solution the distance from $r$ to $t_i$ is at most $L$ for all $i \in [p]$).  As with \textsc{DST} and \textsc{DSN}, \textsc{SLST} has been studied extensively~\cite{kortsarz1997approximating,naor1997improved,hajiaghayi2009approximating,guo2014shallow}.  If we generalize this problem to arbitrary demands, we get the \textsc{Shallow-Light Steiner Network} problem, which is the main problem we study in this paper.  Surprisingly, it has not received nearly the same amount of study (to the best of our knowledge, this paper is the first to consider it explicitly).  It is formally defined as follows (note that we focus on the special case of unit lengths, and will consider general lengths in Sections \ref{sec:approxAlgorithm} and \ref{sec:approxHard}):

\begin{definition}[\textsc{Shallow-Light Steiner Network}]
Given a graph $G=(V,E)$, a cost function $c:E\rightarrow\mathbb{R}^+$, a length function $l:E\rightarrow\mathbb{R}^+$, a distance bound $L$, and $p$ pairs of vertices $\{s_1,t_1\},\mathellipsis,\{s_p,t_p\}$. The objective of \textsc{SLSN} is to find a minimum cost subgraph $G'=(V,S)$, such that for every $i\in[p]$, there is a path between $s_i$ and $t_i$ in $G'$ with length less or equal to $L$.
\end{definition}

Let $H$ be the graph with vertex set $\{s_1,\mathellipsis,s_p,t_1,\mathellipsis,t_p\}$ and edge set $\{\{s_1,t_1\},\mathellipsis,\{s_p,t_p\}\}$. We call $H$ the \emph{demand graph} of the problem. We use $|H|$ to represent the number of edges in $H$.



Both the directed and the length-bounded settings share a dichotomy between considering either star demands (DST/SLST) or totally general demands (DSN/SLSN).  But this gives an obvious set of questions: what demand graphs make the problem ``easy'' (in \textsc{FPT}) and what demand graphs make the problem ``hard'' (\textsc{W}$[1]$-hard)?  Recently, Feldmann and Marx~\cite{feldmann2017complexity} gave a complete characterization for this for \textsc{DSN}. Informally, they proved that if the demand graph is transitively equivalent to an ``almost-caterpillar'' (the union of a constant number of stars where their centers form a path, as well as a constant number of extra edges), then the problem is in \textsc{FPT}, and otherwise the problem is \textsc{W}$[1]$-hard.

While \emph{a priori} there might not seem to be much of a relationship between the directed and the length-bounded problems, there are multiple folklore results that relate them, usually by means of some sort of layered graph.  For example, any \textsc{FPT} algorithm for the \textsc{DST} problem can be turned into an \textsc{FPT} algorithm for \textsc{SLST} (with unit edge lengths) and vice versa through such a reduction (though this is a known result, to the best of our knowledge it has not been written down before, so we include it for completeness in Section \ref{sec:star}).  Such a relationship is not known for more general demands, though.  

In light of these relationships between the directed and the length-bounded settings and the recent results of~\cite{feldmann2017complexity}, it is natural to attempt to characterize the demand graphs that make \textsc{SLSN} easy or hard.  We solve this problem, giving (as in~\cite{feldmann2017complexity}) a complete characterization of easy and hard demand graphs.  Our formal results are given in Section \ref{sec:result}, but informally we show that \textsc{SLSN} is significantly harder than \textsc{DSN}: the only ``easy'' demand graphs are stars (in which case the problem is just \textsc{SLST}) and constant-size graphs.  Even tiny modifications, like a star with a single independent edge, become \textsc{W}$[1]$-hard (despite being in \textsc{FPT} for \textsc{DSN}).  


%

\subsection{Connection to Overlay Routing}
SLSN is particularly interesting due to its connection to overlay routing protocols that use \emph{dissemination graphs} to support next-generation Internet services.  Many emerging applications (such as remote surgery) require extremely low-latency yet highly reliable communication, which the Internet does not natively support.  Babay et al.~\cite{babay2017timely} recently showed that such applications can be supported by using overlay networks to enable routing schemes based on \emph{subgraphs} (dissemination graphs) rather than paths.

Their extensive analysis of real-world data shows that two node-disjoint overlay paths effectively overcome any one fault in the middle of the network, but specialized dissemination graphs are needed to address problems at a flow's source or destination.
Because problems affecting a source typically involve probabilistic loss on that source's outgoing links, a natural approach to increase the probability of a packet being successfully transmitted is to increase the number of outgoing links on which it is sent. In~\cite{babay2017timely}, when a problem is detected at a particular flow's source, that source switches to use a dissemination graph that floods its packets to all of its overlay neighbors and then forwards them from these neighbors to the destination. The paths from the source's neighbors to the destination must meet the application's strict latency requirement, but since the bandwidth used on every edge a packet traverses must be paid for, the total number of edges used should be minimized. Thus, constructing the optimal dissemination graph in this setting is precisely the \textsc{Shallow-Light Steiner Tree} problem, where the root of the demands is the destination and the other endpoints are the neighbors of the source. 

While Babay et al.~\cite{babay2017timely} show that building an optimal \textsc{SLST} is an effective strategy for overcoming failures at either a source or destination, they find that simultaneous failures at both the source \emph{and} the destination of a flow must also be addressed.  Since it is not known in advance which neighbors of the source or destination will be reachable during a failure, the most resilient approach is to require a latency-bounded path from \emph{every} neighbor of the source to \emph{every} neighbor of the destination.  This is precisely \textsc{SLSN} with a complete bipartite demand graph.  Since no \textsc{FPT} algorithm for \textsc{SLSN} with complete bipartite demands was known, \cite{babay2017timely} relied on a heuristic that worked well in practice.  

In the context of dissemination-graph-construction problems, our results provide a good solution for problems affecting either a source or a destination: the \textsc{FPT} algorithm for the \textsc{SLST} problem is quite practical, since overlay topologies typically have bounded degree (and thus a bounded total number of demands). Note that while unit lengths are not typical in overlay networks, handling the true lengths which arise in practice (which are not arbitrary) is a simple modification.  The search for an \textsc{FPT} algorithm for more resilient dissemination graphs (e.g., \textsc{SLSN} with complete bipartite demands) motivated our work, but a trivial corollary of our main results is that this problem is unfortunately \textsc{W}$[1]$-hard.

\section{Our Results and Techniques}\label{sec:result}

In order to distinguish the easy from the hard cases of the \textsc{SLSN} problem with respect to the demand graph, we should first define the problem with respect to a class (set) of demand graphs.

\begin{definition}
Given a class $\mathcal{C}$ of graphs. The problem of \textsc{Shallow-Light Steiner Network} with restricted demand graph class $\mathcal{C}$ ($\textsc{SLSN}_\mathcal{C}$) is the \textsc{SLSN} problem with the additional restriction that the demand graph $H$ of the problem must be isomorphic to some graph in $\mathcal{C}$.
\end{definition}

We define $\mathcal{C}_\lambda$ as the class of all demand graphs with at most $\lambda$ edges, and $\mathcal{C}^*$ as the class of all  star demand graphs (there is a central vertex called the root, and every other vertex in the demand graph is adjacent to the root and only the root). Our main result is that these are \emph{precisely} the easy classes: $\textsc{SLSN}_{\mathcal{C}_\lambda}$ can be solved in polynomial time for fixed $\lambda$, and $\textsc{SLSN}_{\mathcal{C}^*}$ (while \textsc{NP}-hard) is in \textsc{FPT} for parameter $p$. And for any other class $\mathcal{C}$ (i.e., any class which is not just a subset of $\mathcal C^* \cup \mathcal C_{\lambda}$ for some constant $\lambda$), the problem $\textsc{SLSN}_\mathcal{C}$ is \textsc{W}$[1]$-hard with parameter $p$.  Note that $\textsc{SLSN}_{\mathcal{C}^*}$ is precisely the \textsc{SLST} problem, for which a folklore \textsc{FPT} algorithm exists (for completeness, we prove this result in Section \ref{sec:star}).  So our results imply that if we do not have a constant number of demands and are not just \textsc{SLST}, then the problem is actually \textsc{W}$[1]$-hard.

More formally, we prove the following theorems.  

\begin{theorem}\label{thm:const}
For any constant $\lambda>0$, there is a polynomial time algorithm for the unit-length arbitrary-cost $\textsc{SLSN}_{\mathcal{C}_\lambda}$ problem.
\end{theorem}

By ``unit-length arbitrary-cost'' we mean that the length $l(e) = 1$ for all edges $e \in E$, while the cost $c$ is arbitrary. To prove this theorem, we first prove a structural lemma which shows that the optimal solution must be the union of several lowest cost paths with restricted length (these paths may be between steiner nodes, but we show that there cannot be too many). Then we just need to guess all the endpoints of these paths, as well as all the lengths of these paths. It can be proved that there are only $n^{O(p^4)}$ possibilities. Since $p\le\lambda$ is a constant, the running time is polynomial in $n$. \ifprocs\else The algorithm and proof is in Section \ref{sec:const}.\fi

\begin{theorem}\label{thm:star}
The unit-length arbitrary-cost $\textsc{SLSN}_{\mathcal{C}^*}$ problem has an \textsc{FPT} algorithm with parameter $p$.
\end{theorem}

As mentioned, $\textsc{SLSN}_{\mathcal{C}^*}$ is exactly the same as \textsc{SLST}, so we use a folklore reduction between \textsc{SLST} and \textsc{DST} to prove this theorem. The detailed proof is in \ifprocs\else Section \ref{sec:star}. \fi

\begin{theorem}\label{thm:hard}
If $\mathcal{C}$ is a recursively enumerable class, and $\mathcal{C}\nsubseteq\mathcal{C}_\lambda\cup \mathcal{C}^*$ for any constant $\lambda$, then $\textsc{SLSN}_\mathcal{C}$ is \textsc{W}$[1]$-hard with parameter $p$, even in the unit-length and unit-cost case.
\end{theorem}

Many \textsc{W}$[1]$-hardness results for network design problems reduce from the \textsc{Multi-Colored Clique} (\textsc{MCC}) problem, and we are no exception. We reduce from \textsc{MCC} to $\textsc{SLSN}_{\mathcal{C}'}$, where $\mathcal{C}'$ is a specific subset of $\mathcal{C}$ which has some particularly useful properties, and which we show must exist for any such $\mathcal C$. Since $\mathcal{C}'\subseteq\mathcal{C}$, this will imply the theorem.  The reduction is in Section \ref{sec:reduct}.  

All of these results were in the unit-length setting.  We extend both our upper bounds and hardness results to handle arbitrary lengths, but with some extra complications.  If $p=1$ (there is only one demand), then with arbitrary lengths and arbitrary costs the \textsc{SLSN} problem is equivalent to the \textsc{Restricted Shortest Path} problem, which is known to be \textsc{NP}-hard \cite{hassin1992approximation}. Therefore we can no longer hope for a polynomial time exact solution. Note that \textsc{FPT} with constant parameter is equivalent to \textsc{P}, so we change our notion of ``easy'' from ``solvable in \textsc{FPT}'' to ``arbitrarily approximable in \textsc{FPT}'': we show $(1+\varepsilon)$-approximation algorithms for the easy cases, and prove that there is no $\left(\frac{5}{4}-\varepsilon\right)$-approximation algorithm for the hard cases in $f(p)\cdot poly(n)$ time for any function $f$. 

\begin{theorem}\label{thm:approxConst}
For any constant $\lambda>0$, there is a fully polynomial time approximation scheme (\textsc{FPTAS}) for the arbitrary-length arbitrary-cost $\textsc{SLSN}_{\mathcal{C}_\lambda}$ problem.
\end{theorem}

\begin{theorem}\label{thm:approxStar}
There is a $(1+\epsilon)$-approximation algorithm in $O(4^p\cdot poly(\frac{n}{\varepsilon}))$ time for the arbitrary-length arbitrary-cost $\textsc{SLSN}_{\mathcal{C}^*}$ problem.
\end{theorem}

For both upper bounds, we use basically the same algorithm as the unit-length arbitrary-cost case, with some changes inspired by the $(1+\varepsilon)$-approximation algorithm for the \textsc{Restricted Shortest Path} problem \cite{lorenz2001simple}. \ifprocs \else  These results can be found in Section \ref{sec:approxAlgorithm}. \fi

Our next theorem is analogous to Theorem \ref{thm:hard}, but since costs are allowed to be arbitrary we can prove stronger hardness of approximation (under stronger assumptions).

\begin{theorem}\label{thm:approxHard}
Assume that (randomized) Gap-Exponential Time
Hypothesis (Gap-ETH, see \cite{chalermsook2017gap}) holds. Let $\varepsilon>0$ be a small constant, and $\mathcal{C}$ be a recursively enumerable class where $\mathcal{C}\nsubseteq\mathcal{C}_\lambda\cup \mathcal{C}^*$ for any constant $\lambda$. Then, there is no $\left(\frac{5}{4}-\varepsilon\right)$-approximation algorithm in $f(p)\cdot n^{O(1)}$ time for $\textsc{SLSN}_\mathcal{C}$ for any function $f$, even in the unit-length and polynomial-cost case.
\end{theorem}

Note that this theorem uses a much stronger assumption (Gap-ETH rather than \textsc{W}$[1]$ $\neq$ \textsc{FPT}), which assumes that there is no (possibly randomized) algorithm running in $2^{o(n)}$ time can distinguish whether a \textsc{3SAT} formula is satisfiable or at most a $(1-\varepsilon)$-fraction of its clauses can be satisfied. This enables us to utilize the hardness result for a generalized version of the \textsc{MCC} problem from \cite{chitnis2017parameterized}, which will allow us to modify our reduction from Theorem \ref{thm:hard} to get hardness of approximation. \ifprocs \else This result appears in Section \ref{sec:approxHard}. \fi

\subsection{Relationship to~\cite{feldmann2017complexity}}
As mentioned, our results and techniques are strongly motivated and influenced by the work of Feldmann and Marx~\cite{feldmann2017complexity}, who proved similar results in the directed setting.  Informally, they showed that \textsc{Directed Steiner Network} is in \textsc{FPT} if the demand graph is an ``almost-caterpillar'', and otherwise it is \textsc{W}$[1]$-hard.  So they had to show how to reduce from a \textsc{W}$[1]$ problem (MCC, as in our reduction) to \textsc{DSN} where the demand graph is not an almost-caterpillar, and like us had to consider a few different cases depending on the structure of the demand graph.  

The main case of their reduction (which was not already implied by prior work) is when the demand graph is a $2$-by-$k$ complete bipartite graph (i.e., two stars with the same leaf set). For this case, their reduction from \textsc{MCC} uses one star to control the choice of edges in the clique and another star to control the choice of vertices in the clique.  They set this up so that if there is a clique of the right size then the ``edge demands'' and the ``vertex demands'' can be satisfied with low cost by making choices corresponding to the clique, while if no such clique exists then any way of satisfying the two types of demands simultaneously must have larger cost.

The $2$-by-$k$ complete bipartite graph is also a hard demand graph in our setting, and the same reduction from~\cite{feldmann2017complexity} can be straightforwardly modified to prove this (this appears as one of our cases).  However, we prove that far simpler demand graphs are also hard.  Most notably, the ``main'' case of our proof is when the demand graph is a single star together with one extra edge.  Since we have only a single star in our demand graph, we cannot have two ``types'' of demands (vertex demands and edge demands) in our reduction.  So we instead use the star to correspond to ``edge demands'' and use the single extra edge to simultaneously simulate all of the ``vertex demands''.  This makes our reduction significantly more complicated. 

With respect to upper bounds, the algorithm of~\cite{feldmann2017complexity}  is quite complex in part due to the complexity of the demand graphs that it must solve.  Our hardness results for \textsc{SLSN} imply that we need only concern ourselves with demand graphs that are star or have constant size.  The star setting is relatively simple due to a reduction to \textsc{DST}, but it is not obvious how to use any adaptation of~\cite{feldmann2017complexity} (or the earlier~\cite{feldman2006directed}) to handle a constant number of demands for \textsc{SLSN}.  Our algorithm ends up being relatively simple, but requires a structural lemma which was not necessary in the \textsc{DSN} setting.




\section{Algorithms for Unit-Length Arbitrary-Cost \textsc{SLSN}}\label{sec:algorithm}

In this section we discuss the ``easy'' cases of \textsc{SLSN}. We present a
polynomial-time algorithm for \textsc{SLSN} with a constant number of demands
in Section \ref{sec:const}.
In Section \ref{sec:star}, we describe a reduction from \textsc{SLSN} with star
demand graphs to \textsc{DST}, which gives an \textsc{FPT} algorithm. 

\subsection{Constant Number of Demands}\label{sec:const}

For any constant $\lambda$, we show that there is a polynomial-time algorithm
that solves $\textsc{SLSN}_{\mathcal{C}_\lambda}$ (Theorem \ref{thm:const}).
This algorithm relies on the following structural lemma, which allows us to limit the structure of
the optimal solution\ifprocs\else. This lemma works not only for the unit-length case, but also for the arbitrary-length case. \fi

\begin{lemma}\label{lem:intersect}
In any feasible solution $S \subseteq E$ of the \textsc{SLSN} problem, there exists a way to assign a path $P_i$ between $s_i$ and $t_i$ in $S$ for each demand $\{s_i,t_i\}\in H$ such that:
\begin{itemize}
\item For each $i\in[p]$, the total length of $P_i$ is at most $L$ and there is no cycle in $P_i$.
\item For each $i,j\in[p]$ and $u,v\in P_i\cap P_j$, the paths between $u$ and $v$ in $P_i$ and $P_j$ are the same.
\end{itemize}
\end{lemma}
\newcommand{\intersectproof}{
We give a constructive proof.  Let $m=|S|$ and $S=\{e_1,\mathellipsis,e_m\}$.  We first want to modify the lengths to ensure that there is always a unique shortest path. Let $\Delta$ denote the minimum length difference between any two subsets of $S$ with different total length, i.e., 
\begin{equation*}
\Delta=\min_{A,B\subseteq S,\sum_{e\in A}l(e)\ne\sum_{e\in B}l(e)}\left|\sum_{e\in A}l(e)-\sum_{e\in B}l(e)\right|.
\end{equation*}
We create a new length function $g$ where $g(e_i)=l(e_i)+\Delta\cdot2^{-i}$. Note that $\Delta$ is always non-zero for any $S$ which has at least $2$ edges, and the problem is trivial when $|S|=1$.

We now show that any two paths have different lengths under $g$. Consider any two different paths $P_x$ and $P_y$.  If $\sum_{e\in P_x}l(e)\ne\sum_{e\in P_y}l(e)$, then without loss of generality we assume $\sum_{e\in P_x}l(e)<\sum_{e\in P_y}l(e)$. Then
\begin{equation}\label{eqn:keep}
\sum_{e\in P_x}g(e)\le\sum_{e\in P_x}l(e)+\sum_{i=1}^m\Delta\cdot2^{-i}<\sum_{e\in P_x}l(e)+\Delta\le\sum_{e\in P_y}l(e)<\sum_{e\in P_y}g(e).
\end{equation}
Otherwise, if $\sum_{e\in P_x}l(e)=\sum_{e\in P_y}l(e)$, then
\[\sum_{e\in P_x}g(e)-\sum_{e\in P_y}g(e)=\sum_{i:e_i\in P_x}\Delta\cdot2^{-i}-\sum_{i:e_i\in P_y}\Delta\cdot2^{-i}\ne0.\]
Therefore in both cases $P_x$ and $P_y$ have different lengths under $g$.

For each demand $\{s_i,t_i\}\in H$, we let $P_i$ be the shortest path between $s_i$ and $t_i$ in $S$ under the new length function $g$. Because any two paths under $g$ have different length, the shortest path between each $\{s_i,t_i\}\in H$ is unique. In addition, because these are shortest paths and edge lengths are positive, they do not contain any cycles.

For each $i\in[p]$, we can see that $P_i$ is also one of the shortest paths between $s_i$ and $t_i$ under original length function $l$. This is because in equation (\ref{eqn:keep}) we proved that a shorter path under length function $l$ is still a shorter path under length function $g$. Since $S$ is a feasible solution, the shortest path between $s_i$ and $t_i$ in $S$ must have length at most $L$. Thus for each $i\in[p]$, we have $\sum_{e\in P_i}l(e)\le L$. 

For any two different paths $P_i$ and $P_j$, let $u,v \in P_i \cap P_j$.  If the subpath of $P_i$ between $u$ and $v$ is different from the subpath of $P_j$ between $u$ and $v$, then by the uniqueness of shortest paths under $g$ we know that either $P_i$ or $P_j$ is not a shortest path (since one of them could be improved by changing the subpath between $u$ and $v$). This contradicts our definition of $P_i$ and $P_j$, and hence they must use the same subpath between $u$ and $v$. 
}
\ifprocs
\else
\begin{proof}
\intersectproof
\end{proof}
Lemma \ref{lem:intersect} implies that for each two paths $P_i$ and $P_j$, either they do not share any edge, or they share exactly one (maximal) subpath. Since there are only $p$ demands, the total number of shared subpaths is at most $\binom{p}{2}$. Therefore we can solve the unit-length arbitrary-cost $\textsc{SLSN}_{\mathcal{C}_\lambda}$ by guessing these subpaths.

The first step of our algorithm is to guess the endpoints $Q$ of these subpaths, and let $Q'=Q\cup \left(\bigcup_{i=1}^p\{s_i,t_i\}\right)$. The second step is to guess a set $E'\subseteq\{\{u,v\}\mid u,v\in Q',u\ne v\}$. Intuitively, a pair $\{u,v\}\in E'$ means there is a path between $u$ and $v$ in the optimal solution such that only the endpoints of this path is in $Q'$. Then we also guess the length $l'(\{u,v\})$ of such path for each $\{u,v\}\in E'$. Finally, we connect each pair of $u,v\in V$ where $\{u,v\}\in E'$ by lowest cost paths with restricted length $l'(\{u,v\})$, check feasibility, and output the optimal solution. The detailed algorithm is in Algorithm \ref{alg:const} in Section \ref{sec:const}.
\fi

\begin{algorithm}
\caption{Unit-length arbitrary-cost $\textsc{SLSN}_{\mathcal{C}_\lambda}$}\label{alg:const}
\begin{algorithmic}
\State Let $M\gets\sum_{e\in E}c(e)$ and $S\gets E$
\For{$Q\subseteq V$ where $|Q|\le p(p-1)$}
	\State $Q'\gets Q\cup \left(\bigcup_{i=1}^p\{s_i,t_i\}\right)$
	\For{$E'\subseteq\{\{u,v\}\mid u,v\in Q',u\ne v\}$ and $l':E'\rightarrow[L]$}
		\State $T\gets\varnothing$
		\For{$\{u,v\}\in E'$}
			\State $T\gets T\cup\{$the lowest cost path between $u$ and $v$ with length at most $l'(\{u,v\})\}$\\\quad\quad\quad\quad\quad// if such path does not exist, $T$ remains the same
		\EndFor
		\If{$T$ is a feasible solution and $\sum_{e\in T}l'(e)<M$}
			\State $M\gets\sum_{e\in T}c(e)$ and $S\gets T$
		\EndIf
	\EndFor
\EndFor
\\
\Return $S$
\end{algorithmic}
\end{algorithm}

\begin{claim}\label{claim:constTime}
The running time of Algorithm \ref{alg:const} is $n^{O(p^4)}$.
\end{claim}
\begin{proof}
Clearly there are at most $n^{p(p-1)}$ possibilities for $Q$, and for each $Q$ there are at most $2^{(p(p-1)+2p)^2}$ possible sets $E'$ and at most $L^{(p(p-1)+2p)^2}$ possible $l'$. Since we assume unit edge lengths, we can use the Bellman-Ford algorithm to find the lowest cost path within a given length bound in polynomial time. Checking feasibility also takes polynomial time using standard shortest path algorithms. Thus, the running time is at most $n^{p(p-1)}\cdot 2^{(p(p+1))^2}\cdot n^{(p(p+1))^2}\cdot poly(n)$.
\end{proof}

\newcommand{\constproof}{
By Claim \ref{claim:constTime}, the running time of Algorithm \ref{alg:const} is $n^{O(p^4)}$. Since $\lambda$ is constant and $p \leq \lambda$, this running time is polynomial in $n$. Now we will prove correctness. The algorithm always returns a feasible solution, because we replace $S$ by $T$ only if $T$ is feasible, and thus $S$ is always a feasible solution. Therefore, we only need to show that this algorithm returns a solution with cost at most the cost of the optimal solution.

Let the optimal solution be $S^*$. We assign $P_i^*$ for all $i\in[p]$ as in Lemma \ref{lem:intersect}. Recall that path $P_i^*$ and $P_j^*$ can share at most one (maximal) subpath for each $i,j\in[p]$ where $i\ne j$. Let $Q^*$ be the endpoint set of the (maximal) subpaths which are shared by some $P_i^*$ and $P_j^*$, and let $Q'^*=Q^*\cup\bigcup_{i=1}^p\{s_i,t_i\}$.

We can see that the optimal solution $S^*$ can be partitioned to a collection of paths by $Q^*$. We use $E'^*$ to represent whether two vertices in $Q'^*$ are ``adjacent'' on some path $P_i^*$: for any $u,v\in Q'^*$ where $u\ne v$, the set $E'^*$ contains $\{u,v\}$ if and only if there exists $i\in[p]$ such that $u,v\in P_i^*$, and there is no vertex $w\in Q'^*\setminus\{u,v\}$ which is in the subpath between $u$ and $v$ in $P_i^*$. For each $\{u,v\}\in E'^*$, let $P_{\{u,v\}}^*$ be the subpath between $u$ and $v$ on path $P_i^*$. This is well defined because by Lemma \ref{lem:intersect} the subpath is unique. We define $l'^*(\{u,v\})$ as the length of $P_{\{u,v\}}^*$ for each $\{u,v\}\in E'^*$

Note that for any $\{u,v\}\ne\{u',v'\}\in E'^*$, we also know that $P_{\{u,v\}}^*$ and $P_{\{u',v'\}}^*$ are edge-disjoint.  To see this, assume that they do share an edge, and let $u''$ and $v''$ be the endpoints of the (maximal) shared subpath between $P_{\{u,v\}}^*$ and $P_{\{u',v'\}}^*$. Then $u''$ and $v''$ are both in $Q'^*$, and at least one of them is in $Q'^*\setminus\{u,v\}$ or in $Q'^*\setminus\{u',v'\}$, which contradicts our definition of $E'^*$.

Since the algorithm iterates over all possibilities for $Q$, $E'$ and $l'$, there is some iteration in which $Q=Q'^*$, $E'=E'^*$, and $l'\equiv l'^*$. We will show that the algorithm also must find an optimal feasible solution in this iteration.

For each $i\in[p]$, the path $P_i^*$ is partitioned to edge-disjoint subpaths by $Q'^*$. Let $q_i$ be the number of subpaths, and let the endpoints be $s_i=v_{i,0},v_{i,1},\mathellipsis,v_{i,q_i-1},v_{i,q_i}=t_i$. We further let these subpaths be $P_{\{s_i,v_{i,1}\}}^*,P_{\{v_{i,1},v_{i,2}\}}^*,\mathellipsis,P_{\{v_{i,q_i-1},t_i\}}^*$. By the definition of $l'^*$, for each $j\in[q_i]$, there must be a path between $v_{i,j-1}$ and $v_{i,j}$ with length at most $l'^*(\{v_{i,j-1},v_{i,j}\})$ in graph $G$. Thus after the algorithm visited $\{v_{i,j-1},v_{i,j}\}\in E'^*$, the edge set $T$ must contains a path between $u$ and $v$ with length at most $l'^*(\{v_{i,j-1},v_{i,j}\})$. Therefore we know that the edge set $T$ in this iteration contains a path between $s_i$ and $t_i$ with length $\sum_{j=1}^{q_i}l'^*(\{v_{i,j-1},v_{i,j}\})\le L$, and thus it is a feasible solution.

Let $MinCost(u,v,d)$ be the lowest cost for a path between $u$ and $v$ with distance at most $d$ in graph $G$, then the total cost of this solution is $\sum_{\{u,v\}\in E'^*}MinCost(u,v,l'^*(\{u,v\}))$. Moreover, for each $\{u,v\}\in E'^*$ and $\{u',v'\}\in E'^*$ with $\{u,v\}\ne\{u',v'\}$, the paths $P_{\{u,v\}}^*$ and $P_{\{u',v'\}}^*$ are edge-disjoint, and each $P_{\{u,v\}}^*$ has cost at least $MinCost(u,v,l'^*(\{u,v\}))$.  Thus the cost of the optimal solution is at least $\sum_{\{u,v\}\in E'^*}MinCost(u,v,l'^*(\{u,v\}))$, and so the algorithm outputs an optimal solution and it runs in polynomial time.
\qed

\begin{corollary}
For any constant $\lambda>0$, there is a polynomial time algorithm for the arbitrary-length unit-cost $\textsc{SLSN}_{\mathcal{C}_\lambda}$.
\end{corollary}
\begin{proof}
We can use the same technique, but instead of guessing the length $l'$ we guess the cost $c'$, and then find shortest path under cost bound $c'$. We can also use Bellman-Ford algorithm in this step.
\end{proof}
}
\ifprocs
\else
\subsubsection{Proof of Theorem \ref{thm:const}:}
\constproof
\fi

\subsection{Star Demand Graphs ($\textsc{SLSN}_{\mathcal{C}^*}$)}\label{sec:star}

\newcommand{\starproof}{
We do a reduction from $\textsc{SLSN}_{\mathcal{C}^*}$ to the \textsc{DST} problem.  This is essentially folklore. We include it here for completeness.

\begin{definition}[\textsc{Directed Steiner Tree}]
Given a directed graph $G=(V,E)$, a cost function $c:E\rightarrow\mathbb{R}^+$, a root $s$, and $p$ vertices $t_1,\mathellipsis,t_p$, the objective of the \textsc{DST} problem is to find a minimum cost subgraph $G'=(V,S)$, such that for every $i\in[p]$, there is a path from $s$ to $t_i$ in $G'$.
\end{definition}

\begin{theorem}[\cite{feldman2006directed}]\label{thm:DST}
There is an \textsc{FPT} algorithm for the \textsc{DST} problem with parameter $p$.
\end{theorem}


\subsubsection{Proof of Theorem \ref{thm:star}:}
Let $(G=(V,E),c,l \equiv 1,\{\{s_1,t_1\},\mathellipsis,\{s_p,t_p\}\},L)$ be a unit-length arbitrary-cost \textsc{SLSN} instance with restricted demand graph class $\mathcal{C}^*$. Since $\mathcal{C}^*$ is the class of stars, we let $s=s_1=s_2=\mathellipsis=s_p$.

For the reduction, we first create a $(L+1)$-layered graph $G'$, where each layer has $|V|$ vertices. Let $v^{(i)}$ represent the vertex $v\in V$ in layer $i$. Then for each $i \in [L]$ and $u,v\in V$, we add an edge from $u^{(i-1)}$ to $v^{(i)}$ if $\{u,v\} \in E$, and we give this edge cost $c'(u^{(i-1)},v^{(i)})=c(u,v)$. For each $i\in[L]$ and each $v\in V$, we also add an edge $(v^{(i-1)},v^{(i)})$ with cost $c'(v^{(i-1)},v^{(i)})=0$.  

This gives us an instance $(G',c',s^{(0)},t_1^{(L)},\mathellipsis,t_p^{(L)})$ of \textsc{DST}.  Since this reduction clearly takes only polynomial time (since $L \leq n$ due to the unit-length setting), the only thing left is to show that the two instances have the same optimal cost.

Let $S$ be the optimal solution of our starting \textsc{SLSN} instance. Let $d_s(v)$ be the distance between $s$ and $v$ in $S$. We can construct a solution $S'$ for the \textsc{DST} instance of cost at most $c(S)$. First, for each $i\in[L]$ and $\{u,v\}\in S$ with $d_s(u)+1=d_s(v)$, we add $(u^{(d_s(u))},v^{(d_s(v))})$ to $S'$. Then, for each $j\in[p]$ and $i=d_s(t_j),\mathellipsis,L$, we add $(t_j^{(i-1)},t_j^{(i)})$ to $S'$. Note that the cost of $S'$ is at most the cost of $S$, since every non-zero cost edge in $S'$ corresponds to a different edge in $S$ with the same cost. $S'$ is also a feasible solution, because for every $i\in[p]$ there is a path $s$ -- $v_{i,1}$ -- $\mathellipsis$ -- $v_{i,d_s(t_i)-1}$ -- $t_i$ in $S$ with length at most $L$, such that $d_s(v_{i,j})=j$ for each $j\in[d_s(t_i)-1]$, and thus $S'$ contains path $s^{(0)}$ -- $v_{i,1}^{(1)}$ -- $\mathellipsis$ -- $v_{i,d_s(t_i)-1}^{(d_s(t_i)-1)}$ -- $t_i^{(d_s(t_i))}$ -- $\mathellipsis$ -- $t_i^{(L)}$.

Now let $S'$ be the optimal solution of the \textsc{DST} instance. We can construct a solution $S$ for our original $\textsc{SLSN}_{\mathcal{C}^*}$ instance as follows: for any $u,v\in V$ where $u\ne v$, we add $\{u,v\}$ to $S$ if there exists an $i$ such that $(u^{(i-1)},v^{(i)})\in S'$. Clearly the cost of $S$ is at most the cost of $S'$ because every edge in $S$ corresponds to a different edge in $S'$ with the same cost. $S$ is also a feasible solution, since for every $i\in[p]$ there is a path $s^{(0)}=v_{i,0}^{(0)}$ -- $v_{i,1}^{(1)}$ -- $\mathellipsis$ -- $v_{i,L-1}^{(L-1)}$ -- $v_{i,L}^{(L)}=t_i^{(L)}$ in $S'$, and thus $S$ contains path $s$ -- $v_{i,1}$ -- $\mathellipsis$ -- $v_{i,L-1}$ -- $t_i$ with length at most $L$. Notice that there may be $j\in[L]$ such that $v_{i,j}=v_{i,j-1}$, but this only decreases the length and has no effect on cost.

Therefore, the two instances have the same optimal cost. Combining this with Theorem \ref{thm:DST} allows us to get an \textsc{FPT} algorithm for the unit-length arbitrary-cost $\textsc{SLSN}_{\mathcal{C}^*}$ by first reducing to \textsc{DST} and then using the algorithm from Theorem \ref{thm:DST}.\qed
}

\ifprocs
\else
\starproof
\fi

\section{\textsc{W}$[1]$-Hardness for Unit-Length Unit-Cost \textsc{SLSN}}\label{sec:hard}

In this section we prove our main hardness result, Theorem \ref{thm:hard}. We begin with some preliminaries, then give our reduction and proof.

\subsection{Preliminaries}\label{sec:pre}

We prove Theorem \ref{thm:hard} by constructing an \textsc{FPT} reduction from the \textsc{Multi-Colored Clique (MCC)} problem to the unit-length unit-cost $\textsc{SLSN}_\mathcal{C}$ problem for any $\mathcal{C}\nsubseteq\mathcal{C}_\lambda\cup \mathcal{C}^*$.  We begin with the \textsc{MCC} problem.

\begin{definition}[\textsc{Multi-Colored Clique}]
Given a graph $G=(V,E)$, a number $k\in\mathbb{N}$ and a coloring function $c:V\rightarrow[k]$. The objective of the \textsc{MCC} problem is to determine whether there is a clique $T \subseteq V$ in $G$ with $|T| = k$ where $c(x) \neq c(y)$ for all $x,y \in T$.
\end{definition}

For each $i\in[k]$, we define $C_i = \{v \in V : c(v) = i\}$ to be the vertices of color $i$. We can assume that the graph does not contain edges where both endpoints have the same color, since those edges do not affect the existence of a multi-colored clique. It has been proven that the \textsc{MCC} problem is \textsc{W}$[1]$-complete.

\begin{theorem}[\cite{fellows2009parameterized}]\label{thm:mcc}
The \textsc{MCC} problem is \textsc{W}$[1]$-complete with parameter $k$.
\end{theorem}

We first define a few important classes of graphs.  These are the major classes that fall outside of $\mathcal{C}^* \cup\mathcal{C}_{\lambda}$, so we will need to be able to reduce \textsc{MCC} to \textsc{SLSN} where the demand graphs are in these classes, and then this will allow us to can prove the hardness for general $C\nsubseteq\mathcal{C}^* \cup\mathcal{C}_{\lambda}$.  For every $k \in \mathbb{N}$, we define the following graph classes. Each of the first four classes is just one graph up to isomorphism, but classes 5 and 6 are sets of graphs, so we use the notation $\mathcal{H}$ instead of $H$ for these classes. Note that each of the first three classes is just a star with an additional edge, so we use $^*$ to make this clear.  

\begin{enumerate}
\item $H_{k,0}^*$: a star with $k(k-1)$ leaves and an edge with both endpoints \emph{not} in the star.\label{case:disjoint}
\item $H_{k,1}^*$: a star with $(k(k-1)+1)$ leaves and an edge $\{u,v\}$ where $u$ is a leaf of the star and $v$ is not in the star. \label{case:one}
\item $H_{k,2}^*$: a star with $(k(k-1)+2)$ leaves, and an edge $\{u,v\}$ where both $u$ and $v$ are leaves of the star.\label{case:two}
\item $H_{k,k}$: $k(k-1)+1$ edges where all the endpoints are different (i.e., a matching of size $k(k-1)+1$).\label{case:matching}
\item $\mathcal{H}_{2,k}$: the class of graphs that have exactly $k(k-1)+2$ vertices, and contain a $2$ by $k(k-1)$ complete bipartite subgraph (not necessarily an induced subgraph). \label{case:bipartite}
\item $\mathcal{H}_k$: the class of graphs that contain at least one of the graphs in previous five classes as an induced subgraph.\label{case:all}
\end{enumerate}

\ifprocs
\else
We first prove the following lemma.
\fi

\begin{lemma}\label{lem:hk}
For any $k\ge2$, if a graph $H$ is not a star and $H$ has at least $8k^{10}$ edges, then $H\in\mathcal{H}_k$, and we can find an induced subgraph which is isomorphic to a graph in $\{H_{k,0}^*,H_{k,1}^*,H_{k,2}^*,H_{k,k}\} \cup \mathcal{H}_{2,k} \cup \mathcal{H}_k$ in $poly(|H|)$ time. 
\end{lemma}
\newcommand{\hkproof}{
We give a constructive proof. We first claim that either there is a vertex in $H$ which has degree at least $2k^4$ or there is an induced matching in $H$ of size $k^2$.  Suppose that all vertices have degree less than $2k^4$.  Then we can create an induced matching by adding an arbitrary edge $\{u,v\}\in H$ to a edge set $M$, removing all vertices that are adjacent to either $u$ or $v$ from $H$, and repeating until there are no more edges in $H$.  In each iteration we reduce the total number of edges by at most $2\cdot2k^4\cdot2k^4$, thus $|M| \geq \frac{8k^{10}}{8k^8}=k^2$.  Since when we add an edge $\{u,v\}$ we also remove all vertices adjacent to $u$ or $v$, every future edge we add to $M$ will have endpoints which are not adjacent to $u$ or $v$, and thus $M$ is an induced matching of $H$ with size $k^2$.

If $H$ has an induced matching of size $k^2$, then $H\in\mathcal{H}_k$ because it contains $H_{k,k}$ as an induced subgraph, and thus we are done.

Otherwise, $H$ has a vertex $s$ with degree at least $2k^4$.  Let $S$ be the neighbors of $s$. If there is any vertex other than $s$ that is adjacent to at least $k(k-1)$ vertices in $S$, then $H$ contains a $2$ by $k(k-1)$ complete bipartite subgraph, so it contains an induced subgraph $H'\in\mathcal{H}_{2,k}$ and thus is in $\mathcal H_k$.

So suppose that there is no vertex other than $s$ that is adjacent to at least $k(k-1)$ vertices in $S$. Consider the case that there is no edge between any pair of vertices in $S$; then, because $H$ is not a star, there must be an edge $\{u,v\}\in H$ with at least one of $u,v$ not in $S\cup\{s\}$. 
Since both $u$ and $v$ are adjacent to at most $k(k-1)$ vertices in $S$, there are at least $k^4-2\cdot k(k-1)\ge k(k-1)$ vertices in $S$ that are not adjacent to either $u$ or $v$. Let the set of these vertices be $T$. Then the induced subgraph on vertex set $T\cup\{s,u,v\}$ is either $H_{k,0}^*$ or $H_{k,1}^*$, depending on whether $\{u,v\}\cap T$ is an empty set.

Now the only remaining case is that there is at least one edge in $H$ with both endpoints in $S$. In this case, we can find $H_{k,2}^*$ as an induced subgraph as follows: We first let $S_0=S$. Then, in each iteration $t$ we let $v_t$ be a vertex in $S_{t-1}$ that is adjacent to the fewest number of other vertices in $S_{t-1}$. We add $v_t$ to the vertex set $T$, and then delete $v_t$ and all the vertices in $S_{t-1}$ that are adjacent to $v_t$ to get $S_t$. This process repeats until we have $|T|=k(k-1)$.

We can use induction to show that, after each iteration $t\le k(k-1)$, there is always at least one edge in $H$ where both endpoints are in $S_t$. The base case is $t=0$, where such an edge clearly exists. Assume the claim holds for iteration $t-1$, consider the iteration $t\le k(k-1)$. If $v_t$ is not adjacent to any other vertex in $S_{t-1}$, then removing $v_t$ does not affect the fact that there is at least one edge left, and thus the claim still holds. Otherwise, $v_t$ is adjacent to at least one vertex in $S_{t-1}$. Thus, each vertex in $S_{t-1}$ must be adjacent to at least one vertex in $S_{t-1}$. Since there is no vertex other than $s$ which is adjacent to at least $k(k-1)$ vertices in $S$, we know that at most $k^2$ vertices are deleted in each iteration, and thus there are still at least $2k^4-k^2\cdot k(k-1)\ge k^4$ vertices in $S_{t-1}$. Because removing $v_t$ and its neighbors can only affect the degree of at most $k^2(k-1)^2$ vertices in $S_{t-1}$, there must still be an edge left between the vertices in $S_t$.

Let $\{u,v\}$ be one of the edges in $H$ where both endpoints are in $S_t$, then the induced subgraph on vertex set $T\cup\{s,u,v\}$ is $H_{k,2}^*$.  Thus $H\in\mathcal{H}_k$.

It is easy to see that all the previous steps directly find an induced subgraph which is isomorphic to a graph in $\{H_{k,0}^*,H_{k,1}^*,H_{k,2}^*,H_{k,k}\} \cup \mathcal{H}_{2,k} \cup \mathcal{H}_k$ and takes polynomial time, thus the lemma is proved.
}
\ifprocs
\else
\begin{proof}
\hkproof
\end{proof}
\fi

\subsection{Reduction}\label{sec:reduct}

In this subsection, we will prove the following reduction theorem.

\begin{theorem}\label{thm:reduct}
Let $(G=(V,E),c)$ be an \textsc{MCC} instance with parameter $k$, and let $H\in\mathcal{H}_k$ be a demand graph. Then a unit-length unit-cost \textsc{SLSN} instance $(G',L)$ with demand graph $H$ can be constructed in $poly(|V||H|)$ time, and there exists a function $g$ (computable in time $poly(|H|)$) such that the \textsc{MCC} instance has a clique with size $k$ if and only if the \textsc{SLSN} instance has a solution with cost $g(H)$. 
\end{theorem}

In order to prove this theorem, we first introduce a construction for any demand graph $H\in\{H_{k,0}^*,H_{k,1}^*,H_{k,2}^*,H_{k,k}\}\cup\mathcal{H}_{2,k}$, and then use the instances constructed in these cases to construct the instance for general $H\in\mathcal{H}_k$.

The construction for $H\in\mathcal{H}_{2,k}$ is similar to \cite{feldmann2017complexity}, which proves the \textsc{W}$[1]$-hardness of the \textsc{DSN} problem. We change all the directed edges in their construction to undirected, and add some edges and dummy vertices. This construction is presented in \ifprocs\else Section \fi \ref{sec:bipartite}. To handle $H_{k,0}^*$, $H_{k,1}^*$, $H_{k,2}^*$, and $H_{k,k}$, we need to change this basic construction due to the simplicity of the demand graphs. Because the constructions for these four graphs are quite similar, we first introduce the construction for $H_{k,0}^*$ in Section \ref{sec:disjoint}, and then show how to modify it for $H_{k,1}^*$, $H_{k,2}^*$, and $H_{k,k}$ in \ifprocs\else Section \fi \ref{sec:H234}.

\subsubsection{Case \ref{case:disjoint}: $H_{k,0}^*$}\label{sec:disjoint}

Given an \textsc{MCC} instance $(G=(V,E),c)$ with parameter $k$, we create a unit-length and unit-cost \textsc{SLSN} instance $(G',L)$ with demand graph $H_{k,0}^*$ as follows.

We first create a graph $G_k^*$ with integer edge lengths (we will later replace all non-unit length edges by paths). See Figure \ref{fig:G_k^*} for an overview of this graph. The vertex set $V_k^*$ contains $6$ layers of vertices and another group of vertices. The first layer $V_1$ is just a root $r$. The second layer $V_2$ contains a vertex $z_{\{i,j\}}$ for each $1\le i<j\le k$, so there are $\tbinom{k}{2}$ vertices. The third layer $V_3$ contains a vertex $z_{e}$ for each $e\in E$, so there are $|E|$ vertices. The fourth layer $V_4$ contains a vertex $x_{v,j}$ for each $v\in V$ and $j \in [k]$ with $j\ne c(v)$, so there are $|V|\cdot(k-1)$ vertices. The fifth layer $V_5$ again contains a vertex $x_{v,j}'$ for each $v\in V$ and $j \in [k]$ with $j\ne c(v)$. The sixth layer $V_6$ contains a vertex $l_{i,j}$ for each $i,j\in[k]$ where $i\ne j$, so there are $k(k-1)$ vertices. Finally, we have a vertex $y_i$ for $i=0,\mathellipsis,k$, so there are $k+1$ vertices in the set $V_y$.

Let $f_i : \mathbb{N} \rightarrow \mathbb{N}$ be the function defined by $f_i(j) = j+1$ if $j+1 \neq i$ and $f_i(j) = j+2$ if $j+1 = i$. 
This function gives the next integer after $j$, but skips $i$. Let $f_i^t(j)=f_i(f_i(\mathellipsis f_i(j)))$ denote this function repeated $t$ times. Recall that $C_i = \{v \in V : c(v) = i\}$. The edge set $E_k^*$ contains following edges, with lengths as indicated:

\begin{itemize}
\item $E_1=\{\{r,z_{\{i,j\}}\}\mid1\le i<j\le k\}$, each edge in $E_1$ has length $2$.\label{edge:level12}
\item $E_2=\{\{z_{\{c(u),c(v)\}},z_e\}\mid e=\{u,v\}\in E\}$, each edge in $E_2$ has length $1$.\label{edge:level23}
\item $E_3=\{\{z_e,x_{u,c(v)}\}\mid e=\{u,v\}\in E\}$, each edge in $E_3$ has length $2k^2-2$. Note that if $\{z_e,x_{u,c(v)}\}\in E_3$, then $\{z_e,x_{v,c(u)}\}\in E_3$\label{edge:level34}
\item $E_4=\{\{x_{v,j},x_{v,j}'\}\mid v\in V,j\ne c(v)\}$, each edge in $E_4$ has length $1$.\label{edge:level45}
\item $E_5=\{\{x_{v,j}',l_{c(v),j}\}\mid v\in V,j\ne c(v)\}$, each edge in $E_5$ has length $2k^2-2$.\label{edge:level56}
\item $E_{yx}=\{\{y_{i-1},x_{v,f_i(0)}\}\mid i\in[k],v\in C_i\}$, each edge in $E_{yx}$ has length $4$.\label{edge:yx}
\item $E_{xx}=\{\{x_{v,j}',x_{v,f_{c(v)}(j)}\}\mid v\in V,j\in[k]\setminus\{c(v),f_{c(v)}^{k-1}(0)\}\}$, each edge in $E_{xx}$ has length $3$.\label{edge:xx}
\item $E_{xy}=\{\{x_{v,f_i^{k-1}(0)}',y_i\}\mid i\in[k],v\in C_i\}$, each edge in $E_{xy}$ has length $3$.\label{edge:xy}
\end{itemize}

\begin{figure}
\centering
\includegraphics[width=5.5in]{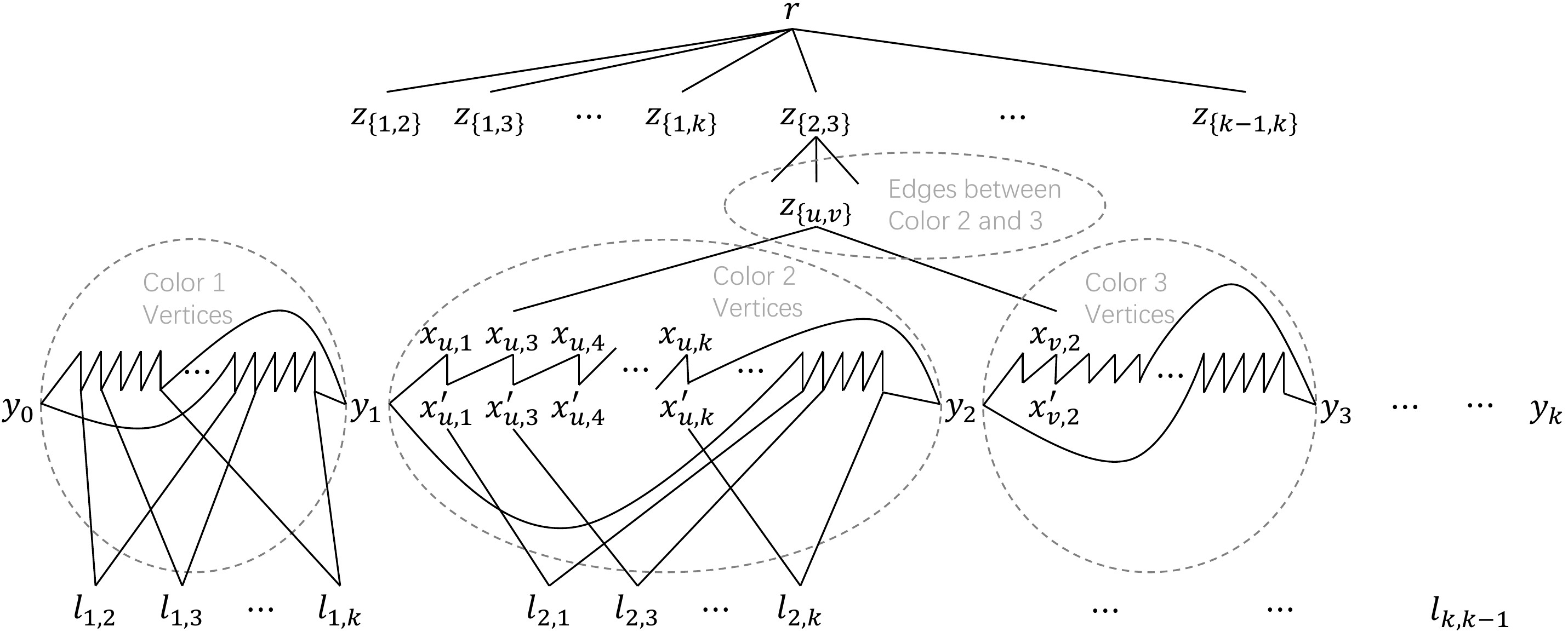}
\caption{$G_k^*$}\label{fig:G_k^*}
\end{figure}

Let $G'$ be the graph obtained from $G_k^*$ by replacing each edge $e\in E_k^*$ by a $length(e)$-hop path.  We create an instance of \textsc{SLSN} on $G'$ by setting the demands to be $\{r,l_{i,j}\}$ for all $i,j\in[k]$ where $i\ne j$, as well as $\{y_0,y_k\}$. Note that these demands form a star with $k(k-1)$ leaves and an edge with both endpoints not in the star, so it is isomorphic to $H_{k,0}^*$. We set the distance bound $L$ to be $4k^2$.

This construction clearly takes $poly(|V||H_{k,0}^*|)$ time. Let $g(H_{k,0}^*)=4k^4-4k^3+\frac{3}{2}k^2+\frac{5}{2}k$, which is clearly computable in $poly(H_{k,0}^*)$ time.  We will first prove the easy direction in the correctness of the construction.

\begin{lemma}\label{lem:disjointCost}
If there is a multi-colored clique of size $k$ in $G$, then there is a solution $S$ for the \textsc{SLSN} instance $(G',L)$ with demand graph $H_{k,0}^*$, and the total cost of $S$ is $g(H_{k,0}^*)$.  
\end{lemma}
\begin{proof}
Let $v_1,\mathellipsis,v_k$ be a multi-colored clique of size $k$ in $G$, where $v_i\in C_i$ for all $i\in[k]$. We create a feasible solution $S$ to our \textsc{SLSN} instance, which contains following paths in $G'$ (i.e., edges in $G_k^*$):
\begin{itemize}
\item $\{r,z_{\{i,j\}}\}$ for each $1\le i<j\le k$. The total cost of these edges is $2\cdot\tbinom{k}{2}=k^2-k$.
\item $\{z_{\{i,j\}},z_{\{v_i,v_j\}}\}$ for each $1\le i<j\le k$. The total cost of these edges is $\tbinom{k}{2}=\frac{k^2-k}{2}$.
\item $\{z_{\{v_i,v_j\}},x_{v_i,j}\}$ and $\{z_{\{v_i,v_j\}},x_{v_j,i}\}$ for each $1\le i<j\le k$. The total cost of these edges is $2\cdot(2k^2-2)\cdot\tbinom{k}{2}=2k^4-2k^3-2k^2+2k$.
\item $\{x_{v_i,j},x_{v_i,j}'\}$ for each $i,j\in[k]$ where $i\ne j$. The total cost of these edges is $2\cdot\tbinom{k}{2}=k^2-k$.
\item $\{x_{v_i,j}',l_{i,j}\}$ for each $i,j\in[k]$ where $i\ne j$. The total cost of these edges is $2\cdot(2k^2-2)\cdot\tbinom{k}{2}=2k^4-2k^3-2k^2+2k$.
\item $\{y_{i-1},x_{v_i,f_i(0)}\}$ for each $i\in[k]$. The total cost of these edges is $4k$.
\item $\{x_{v_i,j}',x_{v_i,f_i(j)}\}$ for each $i\in[k]$ and $j\in[k]\setminus\{i,f_i^{k-1}(0)\}$. The total cost of these edges is $3\cdot k(k-2)=3k^2-6k$.
\item $\{x_{v_i,f_i^{k-1}(0)}',y_i\}$ for each $i\in[k]$. The total cost of these edges is $3k$.
\end{itemize}

Therefore, the total cost is $k^2-k+\frac{k^2-k}{2}+2k^4-2k^3-2k^2+2k+k^2-k+2k^4-2k^3-2k^2+2k+4k+3k^2-6k+3k=4k^4-4k^3+\frac{3}{2}k^2+\frac{5}{2}k = g(H_{k,0}^*)$.

Now we show the feasibility of this solution. For each $i,j\in[k]$ where $i\ne j$, the path between $r$ and $l_{i,j}$ is $r$ -- $z_{\{i,j\}}$ -- $z_{\{v_i,v_j\}}$ -- $x_{v_i,j}$ -- $x_{v_i,j}'$ -- $l_{i,j}$. The length of this path is $2+1+2k^2-2+1+2k^2-2=4k^2$, thus it is a feasible path.

The path between $y_0$ and $y_k$ is $y_0$ -- $x_{v_1,2}$ -- $x_{v_1,2}'$ -- $x_{v_1,3}$ -- $x_{v_1,3}'$ -- $\mathellipsis$ -- $x_{v_1,k}$ -- $x_{v_1,k}'$ -- $y_1$ -- $x_{v_2,1}$ -- $x_{v_2,1}'$ -- $x_{v_2,3}$ -- $x_{v_2,3}'$ -- $\mathellipsis$ -- $y_2$ -- $\mathellipsis$ -- $y_k$. The length of this path is $(4+1\cdot(k-1)+3\cdot(k-2)+3)\cdot k=4k^2$, thus it is a feasible path.
\end{proof}

For the other direction, we begin the proof with a few claims. We first show that the only feasible way to connect $r$ and $l_{i,j}$ is to pick one edge between every two adjacent layers. We can also see in Figure \ref{fig:G_k^*} that for each $i\in[k]$, there are $|C_i|$ disjoint ``zig-zag'' paths between $y_{i-1}$ and $y_i$, and each path corresponds to a vertex with color $i$. We will also show that the only feasible way to connect $y_0$ and $y_k$ is to pick one zig-zag path between each $y_{i-1}$ and $y_i$. From these claims we can then prove that, if the cost of the optimal solution is at most $g(H_{k,0}^*)$, then there is a multi-colored clique in $G$.

\begin{claim}\label{claim:leafPath}
For all $i,j\in[k]$ where $i\ne j$, any path $P_{i,j}$ between $r$ and $l_{i,j}$ with length at most $4k^2$ must be of the form $r$ -- $z_{\{i,j\}}$ -- $z_{\{u,v\}}$ -- $x_{u,j}$ -- $x_{u,j}'$ -- $l_{i,j}$, where $u\in C_i$, $v\in C_j$ and $\{u,v\}\in E$.
\end{claim}
\newcommand{\leafpathproof}{
We can see that $G_k^*$ is a $6$-layer graph with a few additional paths between the fourth layer and the fifth layer. Thus $P_{i,j}$ must contain at least one edge between each two adjacent layers. From the construction of $G_k^*$, all the edges between two adjacent layers have the same length. If we sum up the length from $r$ to the fourth layer plus the length from the fifth layer to $l_{i,j}$, it is already $2+1+2k^2-2+2k^2-2=4k^2-1$. Thus, between the fourth layer and the fifth layer we can only choose one length $1$ edge.

We know that the vertex in the fifth layer must adjacent to $l_{i,j}$, so it must be $x_{u,j}'$ for some $u\in C_i$. Thus, the edge between the fourth layer and the fifth layer must be $\{x_{u,j},x_{u,j}'\}$, because this is the only length $1$ edge adjacent to $x_{u,j}'$. In addition, the only way to go from $r$ to $x_{u,j}$ with one edge per layer is to pass through vertex $z_{\{i,j\}}$ and $z_{\{u,v\}}$ for some $v\in C_j$ and $\{u,v\}\in E$. Therefore $P_{i,j}$ must correspond to an edge $\{u,v\}\in E$ where $u\in C_i$ and $v\in C_j$, and it has form $r$ -- $z_{\{i,j\}}$ -- $z_{\{u,v\}}$ -- $x_{u,j}$ -- $x_{u,j}'$ -- $l_{i,j}$.
}
\ifprocs
\else
\begin{proof}
\leafpathproof
\end{proof}
\fi

\newcommand{\excludeclaim}{
\begin{claim}\label{claim:exclude}
Any path $P_y$ between $y_0$ and $y_k$ with length at most $4k^2$ does not contain any edge in $E_1\cup E_2\cup E_3\cup E_5$.
\end{claim}
\begin{proof}
We prove the claim by contradiction. If $P_y$ contains an edge in $E_1\cup E_2\cup E_3\cup E_5$, it must contain at least two edges with length $2k^2-2$ (one edge to go out of the fourth and the fifth layer, and another one to go back). Since any edge which has endpoint $y_0$ has length $4$ and any edge which has endpoint $y_k$ has length $3$, the total length $2\cdot(2k^2-2)+4+3=4k^2+3$ already exceeds the length bound $4k^2$, giving a contradiction.
\end{proof}
}
\ifprocs
\else
\excludeclaim
\fi

\begin{claim}\label{claim:singlePath}
Any path $P_y$ between $y_0$ and $y_k$ with length at most $4k^2$ can be divided to $k$ subpaths as follows. For each $i\in[k]$, there is a subpath $P_{v_i}$ between $y_{i-1}$ and $y_i$ with length $4k$, of the form $y_{i-1}$ -- $x_{v_i,f_i(0)}$ -- $x_{v_i,f_i(0)}'$ -- $x_{v_i,f_i^2(0)}$ -- $x_{v_i,f_i^2(0)}'$ -- $\mathellipsis$ -- $x_{v_i,f_i^{k-1}(0)}$ -- $x_{v_i,f_i^{k-1}(0)}'$ -- $y_1$, where $v_i\in C_i$.
\end{claim}
\newcommand{\singlepathproof}{
Since we have Claim \ref{claim:exclude}, it suffices to consider the edge set $E_4\cup E_{yx}\cup E_{xx}\cup E_{xy}$. We can see that $E_4\cup E_{yx}\cup E_{xx}\cup E_{xy}$ can be partitioned to $k|V|$ paths, where for each $i\in[k]$ and each $v\in C_i$, there is a path $P_v$ which connects $y_{i-1}$ and $y_i$ with length $4k$. The path is $y_{i-1}$ -- $x_{v,f_i(0)}$ -- $x_{v,f_i(0)}'$ -- $x_{v,f_i^2(0)}$ -- $x_{v,f_i^2(0)}'$ -- $\mathellipsis$ -- $x_{v,f_i^{k-1}(0)}$ -- $x_{v,f_i^{k-1}(0)}'$ -- $y_1$. We can see that these paths are vertex disjoint except for the endpoints $y_0,y_1,\mathellipsis,y_k$.

Therefore, the only way to go from $y_0$ to $y_k$ is by passing through $y_0,y_1,\mathellipsis,y_k$ one-by-one. Thus, for each $i\in[k]$, $P_y$ must contain a subpath $P_{v_i}$ where $v_i\in C_i$. Because each of these subpaths has length $4k$, the total cost is already $4k\cdot k=4k^2$, which is exactly the length bound. Therefore, $P_y$ can not contain any other edge, which proves the lemma.
}
\ifprocs
\else
\begin{proof}
\singlepathproof
\end{proof}
\fi

Now, we can prove the other direction in the correctness of the construction.

\begin{lemma}\label{lem:disjointClique}
Let $S$ be an optimal solution for the \textsc{SLSN} instance $(G',L)$ with demand graph $H_{k,0}^*$. If $S$ has cost at most $ g(H_{k,0}^*) = 4k^4-4k^3+\frac{3}{2}k^2+\frac{5}{2}k$, then there is a multi-colored clique of size $k$ in $G$.
\end{lemma}
\begin{proof}
For each $i,j\in[k]$ with $i\ne j$, let $P_{i,j}$ be a (arbitrarily chosen) path in $S$ which connects $r$ and $l_{i,j}$ with length at most $L=4k^2$. Let $\mathcal{P}=\{P_{i,j}\mid i,j\in[k],i\ne j\}$ be the set of all these paths. We also let $P_y$ be a (arbitrary) path in $S$ of length at most $L$ which connects $y_0$ and $y_k$.

From Claim \ref{claim:singlePath}, $P_y$ can be divided to $k$ subpaths, each of which corresponds to a vertex $v_i$. We will show that $v_1,\mathellipsis,v_k$ form a clique in $G$ (i.e., for each $1\le i<j\le k$, we have $\{v_i,v_j\}\in E$).

We first prove that these paths must share certain edges due to the cost bound of the optimal solution. From Claim \ref{claim:leafPath}, we know that each $P_{i,j}$ costs exactly $2+1+2k^2-2+1+2k^2-2=4k^2$. In addition, from the form of $P_{i,j}$ we can also see that these paths are almost disjoint, except that $P_{i,j}$ and $P_{j,i}$ may share a length $2$ edge $\{r,z_{\{i,j\}}\}\in E_1$ and a length $1$ edge $\{z_{\{i,j\}},z_e\}\in E_2$. Therefore, in order to satisfy the demands between $r$ and all of the $l_{i,j}$'s, the total cost of the edges in $S\cap(E_1\cup E_2\cup E_3\cup E_4\cup E_5)$ is at least $4k^2\cdot k(k-1)-\tbinom{k}{2}\cdot(2+1)=4k^4-4k^3-\frac{3}{2}k^2+\frac{3}{2}k$, even if every $P_{i,j}$ and $P_{j,i}$ do share edge $\{r,z_{\{i,j\}}\}$ and edge $\{z_{\{i,j\}},z_e\}$.

We now calculate the cost of the edges in $S\cap(E_{yx}\cup E_{xx}\cup E_{xy})$. From Claim \ref{claim:singlePath}, the total cost of edges in $P_y\cap(E_{yx}\cup E_{xx}\cup E_{xy})$ is at least $(4+3\cdot(k-1)+3)\cdot k=3k^2+k$. Thus, the total cost is already at least $\left(4k^4-4k^3-\frac{3}{2}k^2+\frac{3}{2}k\right)+(3k^2+k)=4k^4-4k^3+\frac{3}{2}k^2+\frac{5}{2}k=g(H_{k,0}^*)$, so $S$ cannot contain any edge which has not been counted yet.

Therefore, every edge in $P_y\cap E_4$ must appear in some path in $\mathcal{P}$. In fact, by the form of the paths in $\mathcal{P}$, we can see that for each $i,j\in[k]$ where $i\ne j$, the edge $\{x_{v_i,j},x_{v_i,j}'\}\in P_y\cap E_4$ can only appear in path $P_{i,j}$, rather than any other $P_{i',j'}$. Thus $x_{v_i,j}$ is in path $P_{i,j}$, and similarly $x_{v_j,i}$ is in path $P_{j,i}$. Recall that $P_{i,j}$ and $P_{j,i}$ must share an edge $\{z_{\{i,j\}},z_e\}$ for some $e\in E$ because of the cost bound, and $z_{\{v_i,v_j\}}$ is the only vertex which adjacent to both $x_{v_i,j}$ and $x_{v_j,i}$, we can see that $e$ can only be $\{v_i,v_j\}$. Therefore $\{v_i,v_j\}\in E$, which proves the lemma.
\end{proof}

\newcommand{\casess}{
\ifprocs
\else
\subsubsection{Case \ref{case:one}, \ref{case:two}, and \ref{case:matching}:}\label{sec:H234}
\fi

Cases \ref{case:one}, \ref{case:two}, and \ref{case:matching} are basically the same as Case \ref{case:disjoint}, so we discuss them in the same subsection.

\text{}\\
\textbf{Case \ref{case:one}}: $H_{k,1}^*$

We use the same $G_k^*$, $G'$, and $L$ in the construction of the \textsc{SLSN} instance for demand graph $H_{k,0}^*$, and also set $g(H_{k,1}^*)=4k^4-4k^3+\frac{3}{2}k^2+\frac{5}{2}k$. The only difference is the demand graph. Besides the demand of $\{r,l_{i,j}\}$ for all $i,j\in[k]$ where $i\ne j$, and $\{y_0,y_k\}$, there is a new demand $\{r,y_0\}$. Clearly this new demand graph is a star with $(k(k-1)+1)$ leaves, and an edge in which exactly one of the endpoints is a leaf of the star, so it is isomorphic to $H_{k,1}^*$.

Assume there is a multi-colored clique of size $k$ in $G$. The paths connecting previous demands in the solution of the \textsc{SLSN} instance are the same as Case \ref{case:disjoint}. The path between $r$ and $y_0$ is $r$ -- $z_{\{1,2\}}$ -- $z_{\{v_1,v_2\}}$ -- $x_{v_1,2}$ -- $y_0$. All the edges in this path are already in the previous paths, so the cost remains the same. The length of this path is $2+1+2k^2-2+4=2k^2+5<4k^2$, which satisfies the length bound.

Assume there is a solution for the \textsc{SLSN} instance $(G',L,H_{k,1}^*)$ with cost $4k^4-4k^3+\frac{3}{2}k^2+\frac{5}{2}k$. The proof that there exists a multi-colored clique of size $k$ in $G$ is the same as Case \ref{case:disjoint}.

\text{}\\
\textbf{Case \ref{case:two}}: $H_{k,2}^*$

As in Case \ref{case:one}, only the demand graph changes. The new demand graph is the same as in Case \ref{case:one} but again with a new demand $\{r,y_k\}$. Since $\{r, y_0\}$ was already a demand, our new demand graph is a star with $(k(k-1)+2)$ leaves (the $l_{i,j}$'s and $y_0$ and $y_k$), and an edge between two of its leaves ($y_0$ and $y_k$), which is isomorphic to $H_{k,2}^*$.

Assume there is a multi-colored clique of size $k$ in $G$. The paths connecting previous demands in the solution of the \textsc{SLSN} instance are the same as Case \ref{case:one}. The path between $r$ and $y_k$ is $r$ -- $z_{\{k-1,k\}}$ -- $z_{\{v_{k-1},v_k\}}$ -- $x_{v_k,k-1}$ -- $y_k$. All the edges in this path are already in the previous paths, so the cost stays the same. The length of this path is $2+1+2k^2-2+4=2k^2+5<4k^2$, which satisfies the length bound.

Assume there is a solution for the \textsc{SLSN} instance $(G',L,H_{k,2}^*)$ with cost $4k^4-4k^3+\frac{3}{2}k^2+\frac{5}{2}k$. The proof that there exists a multi-colored clique of size $k$ in $G$ is the same as Case \ref{case:disjoint}.

\text{}\\
\textbf{Case \ref{case:matching}}: $H_{k,k}$

In order to get $H_{k,k}$ as our demand graph, we have to slightly change the construction from Case \ref{case:disjoint}. We still first make a weighted graph $G_{k,k}=(V_{k,k},E_{k,k})$ and then transform it to the unit-length unit-cost graph $G'$. For the vertex set $V_{k,k}$, we add another layer of vertices $V_0=\{l_{i,j}'\mid i,j\in[k],i\ne j\}$ to $V_k^*$ before the first layer $V_1$. For the edge set $E_{k,k}$, we include all the edges in $E_k^*$, but change the length of edges in $E_1$ to length $1$. We also add another edge set $E_0=\{\{l_{i,j}',r\}\mid i,j\in[k],i\ne j\}$. Each edge in $E_0$ has length $1$.

The demands are $\{l_{i,j}',l_{i,j}\}$ for each $i,j\in[k]$ where $i\ne j$, as well as $\{y_0,y_k\}$. This is a matching of size $k(k-1)+1$, which is isomorphic to $H_{k,k}$. We still set the length bound to be $L=4k^2$, and set $g(H_{k,k})=4k^4-4k^3+2k^2+2k$.

If there is a multi-colored clique of size $k$ in $G$, the construction for the solution in $G'$ is similar to Case \ref{case:disjoint}. For each $i,j\in[k]$ where $i\ne j$, the path between $l_{i,j}'$ and $l_{i,j}$ becomes $l_{i,j}'$ -- $r$ -- $z_{\{i,j\}}$ -- $z_{\{v_i,v_j\}}$ -- $x_{v_i,j}$ -- $x_{v_i,j}'$ -- $l_{i,j}$ (i.e., one more layer before the root $r$).  It is easy to see that the length bound and size bound are still satisfied.

Assume there is a solution for the \textsc{SLSN} instance $(G',L,H_{k,k})$ with cost $4k^4-4k^3+2k^2+2k$. The proof that there exists a multi-colored clique of size $k$ in $G$ is essentially the same as Case \ref{case:disjoint}, except the path between $l_{i,j}'$ and $l_{i,j}$ has one more layer.

\ifprocs
\else
\subsubsection{Case \ref{case:bipartite}: $\mathcal{H}_{2,k}$}\label{sec:bipartite}
\fi

In this case, we slightly modify the reduction of \cite{feldmann2017complexity}. We first change all the edges from directed to undirected. In addition, in \cite{feldmann2017complexity} the demand graph is precisely a $2$-by-$k(k-1)$ bipartite graph, but we also handle the generalization in which there may be more demands between vertices on each sides (i.e., the $2$-by-$k(k-1)$ bipartite graph is just a subgraph of our demands). In order to do this, we add some dummy vertices and some edges.

Given an \textsc{MCC} instance $(G=(V,E),c)$ with parameter $k$, and a demand graph $H\in \mathcal{H}_{2,k}$, we create a unit-length and unit-cost \textsc{SLSN} instance $G'$ with demand isomorphic to $H$ as follows.

We first create a weighted graph $G_{2,k}=(V_{2,k},E_{2,k})$. The vertex set $V_{2,k}$ contains $5$ layers of vertices. The first layer $V_1$ is just two roots $r_1,r_2$. The second layer $V_2$ contains a vertex $z_{\{i,j\}}$ for each $1\le i<j\le k$, and a vertex $y_i$ for each $i\in[k]$. The third layer $V_3$ contains a vertex $z_{e}$ for each $e\in E$, and a vertex $y_v$ for each $v\in V$. The fourth layer $V_4$ contains a vertex $x_{v,j}$ for each $v\in V$ and $j\ne c(v)$. The fifth layer $V_5$ contains a vertex $l_{i,j}$ for each $i,j\in[k]$ where $i\ne j$.

The edge set $E_{2,k}$ contains the following edges:
\begin{itemize}
\item $E_{11}=\{\{r_1,z_{\{i,j\}}\},1\le i<j\le k\}$, each edge in $E_{11}$ has length $1$.\label{edge:z12}
\item $E_{12}=\{\{z_{\{c(u),c(v)\}},z_e\}\mid e=\{u,v\}\in E\}$, each edge in $E_{12}$ has length $1$.\label{edge:z23}
\item $E_{13}=\{\{z_e,x_{u,c(v)}\}\mid e=\{u,v\}\in E\}$, each edge in $E_{13}$ has length $1$. Note that if $\{z_e,x_{u,c(v)}\}\in E_{13}$, then $\{z_e,x_{v,c(u)}\}\in E_{13}$\label{edge:z34}
\item $E_{21}=\{\{r_2,y_i\}\mid i\in[k]\}$, each edge in $E_{21}$ has length $1$.\label{edge:y12}
\item $E_{22}=\{\{y_{c(v)},y_v\}\mid v\in V\}$, each edge in $E_{22}$ has length $1$.\label{edge:y23}
\item $E_{23}=\{\{y_v,x_{v,j}\}\mid v\in V,j\ne c(v)\}$, each edge in $E_{23}$ has length $1$.\label{edge:y34}
\item $E_{xl}=\{\{x_{v,j},l_{c(v),j}\}\mid v\in V,j\ne c(v)\}$, each edge in $E_{xl}$ has length $4$.\label{edge:xl}
\item $E_{ll}=\{\{l_{i,j},l_{i',j'}\}\mid i,j,i',j'\in[k],i\ne j,i'\ne j',(i,j)\ne(i',j')\}$, each edge in $E_{ll}$ has length $7$.\label{edge:l}
\end{itemize}

\begin{figure}
\centering
\includegraphics[width=5.5in]{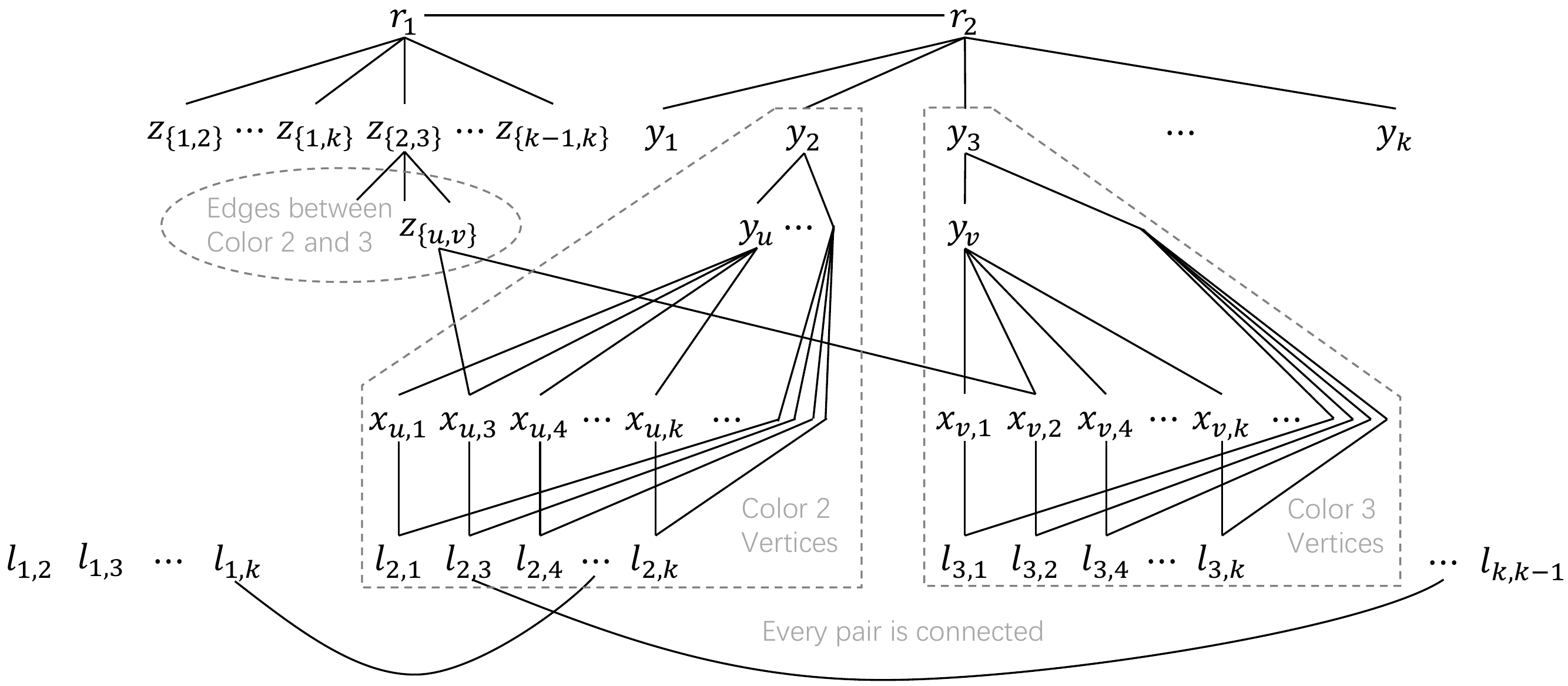}
\caption{$G_{2,k}$}\label{fig:G2k}
\end{figure}

We get a unit-length graph $G'$ from $G_{2,k}$ by replacing every edge $e\in E_{2,k}$ by a $length(e)$-hop path.  Our \textsc{SLSN} instance consists of the graph $G'$, length bound $L = 7$, and the following demands (which will be isomorphic to $H$).  For each $r \in\{r_1,r_2\}$ and $i,j \in [k]$ with $i \neq j$, there is a demand between $r$ and $l_{i,j}$ (note that these demands form a $2$ by $k(k-1)$ complete bipartite graph. Let this complete bipartite subgraph be $B$. For the rest of the demands, we arbitrarily choose a mapping between $V_1=\{r_1,r_2\}$ and the $2$-side of the bipartite graph in $H$, as well as a mapping between $V_5=\{l_{i,j}\mid i,j\in[k],i\ne j\}$ and the $k(k-1)$-side. There is a demand between two vertices $u,v\in V_1\cup V_5$ if there is an edge between $u,v$ in $H$.

This construction clearly takes $poly(|V||H|)$ time. Let $g(H)=7|H|-7k^2+9k-7\cdot\mathds{1}_{\{r_1,r_2\}\in H}$, where $\mathds{1}_{\{r_1,r_2\}\in H}$ is an indicator variable for $\{r_1, r_2\}$ being a demand in $H$. This function is also computable in time $poly(|H|)$. We first prove the easy direction in the correctness of the reduction.

\begin{lemma}\label{lem:bipartiteCost}
If there is a multi-colored clique of size $k$ in $G$, then there is a solution $S$ for the \textsc{SLSN} instance $(G',L)$ with demand graph $H\in\mathcal{H}_{2,k}$, and the total cost of $S$ is $7|H|-7k^2+9k-7\cdot\mathds{1}_{\{r_1,r_2\}\in H}$.
\end{lemma}
\begin{proof}
Let $v_1,\mathellipsis,v_k$ be a multi-colored clique of size $k$ in $G$, where $v_i\in C_i$ for all $i\in[k]$. We create a feasible solution $S$ to our \textsc{SLSN} instance, which contains following paths in $G'$ (i.e., edges in $G_{2,k}$):

\begin{itemize}
\item $\{r_1,z_{\{i,j\}}\}$ for each $1\le i<j\le k$. The total cost of these edges is $\tbinom{k}{2}=\frac{k^2-k}{2}$.
\item $\{z_{\{i,j\}},z_{\{v_i,v_j\}}\}$ for each $1\le i<j\le k$. The total cost of these edges is $\tbinom{k}{2}=\frac{k^2-k}{2}$.
\item $\{z_{\{v_i,v_j\}},x_{v_i,j}\}$ and $\{z_{\{v_i,v_j\}},x_{v_j,i}\}$ for each $1\le i<j\le k$. The total cost of these edges is $2\cdot\tbinom{k}{2}=k^2-k$.
\item $\{r_2,y_i\}$ for each $i\in[k]$. The total cost of these edges is $k$.
\item $\{y_i,y_{v_i}\}$ for each $i\in[k]$. The total cost of these edges is $k$.
\item $\{y_{v_i},x_{v_i,j}\}$ for each $i,j\in[k]$ where $i\ne j$. The total cost of these edges is $2\cdot\tbinom{k}{2}=k^2-k$.
\item $\{x_{v_i,j},l_{i,j}\}$ for each $i,j\in[k]$ where $i\ne j$. The total cost of these edges is $4\cdot2\cdot\tbinom{k}{2}=4k^2-4k$.
\item $\{u,v\}$ for each $\{u,v\}\in H\setminus(B\cup\{\{r_1,r_2\}\})$. The total cost of these edges is $7\cdot(|H|-2\cdot k(k-1)-\mathds{1}_{\{r_1,r_2\}\in H})=7|H|-14k^2+14k-7\cdot\mathds{1}_{\{r_1,r_2\}\in H}$.
\end{itemize}

Therefore, the total cost is $\frac{k^2-k}{2}+\frac{k^2-k}{2}+k^2-k+k+k+k^2-k+4k^2-4k+7|H|-14k^2+14k-7\cdot\mathds{1}_{\{r_1,r_2\}\in H}=7|H|-7k^2+9k-7\cdot\mathds{1}_{\{r_1,r_2\}\in H}$.

Now we show the feasibility of this solution. For each $i,j\in[k]$ where $i\ne j$, the path between $r_1$ and $l_{i,j}$ is $r_1$ -- $z_{\{i,j\}}$ -- $z_{\{v_i,v_j\}}$ -- $x_{v_i,j}$ -- $l_{i,j}$, and the path between $r_2$ and $l_{i,j}$ is $r_2$ -- $y_i$ -- $y_{v_i}$ -- $x_{v_i,j}$ -- $l_{i,j}$. Both paths have length $7$, which is within the length bound. For each $\{u,v\}\in H\setminus(B\cup\{\{r_1,r_2\}\})$, $u$ and $v$ have an edge with length $7$, thus a path under the length bound exists. Finally, if there exists a demand between $r_1$ and $r_2$, we can follow the path $r_1$ -- $z_{\{1,2\}}$ -- $z_{\{v_1,v_2\}}$ -- $x_{v_1,2}$ -- $y_{v_1}$ -- $y_1$ -- $r_2$, which has length $6$.
\end{proof}

Now we prove the other direction.

Let $S$ be an optimal solution for the \textsc{SLSN} instance $(G',L)$ with demand graph $H_{k,0}^*$. If $S$ has cost at most $4k^4-4k^3+\frac{3}{2}k^2+\frac{5}{2}k$, then there is a multi-colored clique of size $k$ in $G$.

\begin{lemma}\label{lem:bipartiteClique}
Let $S$ be an optimal solution for the \textsc{SLSN} instance $(G',L)$ with demand graph $H\in\mathcal{H}_{2,k}$. If $S$ has cost at most $7|H|-7k^2+9k-7\cdot\mathds{1}_{\{r_1,r_2\}\in H}$, then there is a multi-colored clique of size $k$ in $G$.
\end{lemma}
\begin{proof}
For each $i,j\in[k]$ where $i\ne j$, let $P_{1,i,j}\subseteq S$ be a (arbitrarily chosen) path between $r_1$ and $l_{i,j}$ with length at most $7$, and $P_{2,i,j}\subseteq S$ be a (arbitrarily chosen) path between $r_2$ and $l_{i,j}$ with length at most $7$. Let $\mathcal{P}_1=\{P_{1,i,j}\mid i,j\in[k],i\ne j\}$, and $\mathcal{P}_2=\{P_{2,i,j}\mid i,j\in[k],i\ne j\}$. As in lemma \ref{lem:disjointClique}, we first show that some edges must be shared by multiple paths by calculating the total cost.

In order to satisfy the demand for each $\{l_{i,j},l_{i',j'}\}\in H\setminus(B\cup\{\{r_1,r_2\}\})$, the only way is to use the edge between $l_{i,j}$ and $l_{i',j'}$ in $E_{ll}$. Otherwise, suppose the path has more than one edge, since the only edges incident on any $l_{i,j}$ have length either $4$ or $7$, the cost of two of these edges already exceeds the length bound. Thus the total cost of the edges in $S\cap E_{ll}$ is at least $7|H|-7|B|-7\cdot\mathds{1}_{\{r_1,r_2\}\in H}=7|H|-14k^2+14k-7\cdot\mathds{1}_{\{r_1,r_2\}\in H}$.

We can see that each of the paths in $\mathcal{P}_1\cup\mathcal{P}_2$ must have exactly one edge between every two adjacent levels, and they cannot have any other edges because of the length bound. Thus, each path $P_{1,i,j}\in\mathcal{P}_1$ must have form $r_1$ -- $z_{\{i,j\}}$ -- $z_{\{u,v\}}$ -- $x_{u,j}$ -- $l_{i,j}$ for some $\{u,v\}\in E$ with $u\in C_i$ and $v\in C_j$, and each path in $P_{2,i,j}\in\mathcal{P}_2$ must have form $r_2$ -- $y_i$ -- $y_{v}$ -- $x_{v,j}$ -- $l_{i,j}$ for some $v\in C_i$.

By looking at the form of paths in $\mathcal{P}_1$, we can see that these paths are almost disjoint, except that $P_{1,i,j}$ and $P_{1,j,i}$ may share edge $\{r_1,z_{\{i,j\}}\}\in E_{11}$ and edge $\{z_{\{i,j\}},z_e\}\in E_{12}$. Since paths in $\mathcal{P}_1$ only contain edges in $E_{11}\cup E_{12}\cup E_{13}\cup E_{xl}$, the cost of edges in $S\cap(E_{11}\cup E_{12}\cup E_{13}\cup E_{xl})$ must be at least $7\cdot k(k-1)-\tbinom{k}{2}-\tbinom{k}{2}=6k^2-6k$, even if every $P_{1,i,j}$ and $P_{1,j,i}$ do share edge $\{r_1,z_{\{i,j\}}\}$ and edge $\{z_{\{i,j\}},z_e\}$.

We then look at the form of paths in $\mathcal{P}_2$.  We can see that the first $3$ hops of these paths only contain edges in $E_{21}\cup E_{22}\cup E_{23}$. In addition, these paths are all disjoint on edges in $E_{23}$. Moreover, in order to reach all $l_{i,j}$ from $r_2$ within length $7$, these paths should contain all edges in $E_{21}$ and at least $k$ edges in $E_{22}$. Therefore, the total cost of edges in $S\cap(E_{21}\cup E_{22}\cup E_{23})$ should be at least $k(k-1)+k+k=k^2+k$.

By summing up all these edges, the total cost of edges in $S$ is already at least $7|H|-14k^2+14k-7\cdot\mathds{1}_{\{r_1,r_2\}\in H}+6k^2-6k+k^2+k=7|H|-7k^2+9k-7\cdot\mathds{1}_{\{r_1,r_2\}\in H}=g(H)$, which means $S$ cannot contain any edge that has not been counted before.

Therefore, $S$ must contain exactly $k$ edges in $E_{22}$, and each of these edges must have a different $y_i$ as an endpoint. We let these edges be $\{y_1,y_{v_1}\},\mathellipsis,\{y_k,y_{v_k}\}$, where $v_i\in C_i$ for all $i\in[k]$. We claim that $v_1,\mathellipsis,v_k$ forms a (multicolored) clique in $G$.

For each $1\le i<j\le k$, by looking at the form of paths in $\mathcal{P}_2$, we know that the path $P_{2,i,j}$ must be $r_2$ -- $y_i$ -- $y_{v_i}$ -- $x_{v_i,j}$ -- $l_{i,j}$. Because of the total cost limitation, the edge $\{x_{v_i,j},l_{i,j}\}\in P_{2,i,j}\cap E_{xl}$ must also appear in some path in $\mathcal{P}_1$. By looking at the form of the paths in $\mathcal{P}_1$, the only possible path is $P_{1,i,j}$. Similarly, path $P_{2,j,i}$ must share edge $\{x_{v_j,i},l_{j,i}\}$ with $P_{1,j,i}$. Again by looking at the form of the paths in $\mathcal{P}_1$, the edge in $\{z_{\{i,j\}},z_e\}\in S\cap E_{12}$ which is shared by $P_{1,i,j}$ and $P_{1,j,i}$ must have $e=\{v_i,v_j\}$, which means $\{v_i,v_j\}\in E$.

Therefore, $v_1,\mathellipsis,v_k$ forms a clique in $G$.
\end{proof}
}
\ifprocs
\else
\casess
\fi

\subsubsection{Case \ref{case:all}: $\mathcal{H}_k$}\label{sec:all}

We now want to construct an \textsc{SLSN} instance for a demand graph $H\in\mathcal{H}_k$ from an \textsc{MCC} instance $(G=(V,E),c)$ with parameter $k$.  By the definition of $\mathcal H_k$, for some $t\in[5]$ there is a graph $H^{(t)}$ of Case $t$ that is an induced subgraph of $H$. We use Lemma \ref{lem:hk} to find the graph $H^{(t)}$. Let $(G^{(t)},L)$ be the \textsc{SLSN} instance obtained by applying our reduction for Case $t$ to the \textsc{MCC} instance $(G, c)$, and let the corresponding function be $g^{(t)}$.

We now want to transform the \textsc{SLSN} instance $(G^{(t)}, L)$ with demand graph $H^{(t)}$ into a new \textsc{SLSN} instance $(G',L)$ with demand graph $H$, so that instance $(G^{(t)},L,H^{(t)})$ has a solution with cost $g^{(t)}(H^{(t)})$ if and only if instance $(G',L,H)$ has a solution with cost $g(H)=g^{(t)}(H^{(t)})+L\cdot(|H|-|H^{(t)}|)$. If there is such a construction which runs in polynomial time, then there is a multi-colored clique of size $k$ in $G$ if and only if instance $(G',L,H)$ has a solution with cost $g(H)$. This will then imply Theorem \ref{thm:reduct}.

The graph $G'$ is basically just $G^{(t)}$ with some additional vertices and edges from $H\setminus H^{(t)}$. For each vertex $v$ in $H$ but not in $H^{(t)}$, we add a new vertex $v$ to $G'$. For each edge $\{u,v\}\in H\setminus H^{(t)}$, we add an $L$-hop path between $u$ and $v$ to $G'$.

The construction still takes $poly(|V||H|)$ time, because the construction for the previous cases takes $poly(|V||H^{(t)}|)$ time and the construction for Case \ref{case:all} takes $poly(|G^{(t)}||H|)$ time. Here $|H^{(t)}|\le|H|$, and we know that $|G^{(t)}|$ is polynomial in $|V|$ and $|H^{(t)}|$.

\begin{lemma}
\textsc{SLSN} instance $(G^{(t)},L,H^{(t)})$ has a solution with cost $g^{(t)}(H^{(t)})$ if and only if instance $(G',L,H)$ has a solution with cost $g(H)=g^{(t)}(H^{(t)})+L\cdot(|H|-|H^{(t)}|)$.
\end{lemma}
\begin{proof}
If instance $(G^{(t)},L,H^{(t)})$ has a solution with cost $g^{(t)}(H^{(t)})$. Let the solution be $S^{(t)}$. For each $e=\{u,v\}\in H\setminus H^{(t)}$, let the new $L$-hop path between $u$ and $v$ in $G'$ be $P_e$. Then $S^{(t)}\cup\bigcup_{e\in H\setminus H^{(t)}}P_e$ is a solution to $G'$ with cost $g^{(t)}(H^{(t)})+L\cdot(|H|-|H^{(t)}|)$.

If instance $(G',L,H)$ has a solution with cost $g^{(t)}(H^{(t)})+L\cdot(|H|-|H^{(t)}|)$, let the solution be $S$. Since for each $e=\{u,v\}\in H\setminus H^{(t)}$, the only path between $u$ and $v$ in $G'$ within the length bound is the new $L$-hop path $P_e$, any valid solution must include all these $P_e$, which has total cost $L\cdot(|H|-|H^{(t)})$. In addition, for each demand $\{u,v\}$ which is also in $H$, any path between $u$ and $v$ in $G'$ within the length bound will not include any new edge, because otherwise it will strictly contain an $L$-hop path, and have length more than $L$. Therefore, $S\setminus\bigcup_{e\in H\setminus H^{(t)}}P_e$ is a solution to $G^{(t)}$ with cost $g^{(t)}(H^{(t)})$.
\end{proof}

Therefore Theorem \ref{thm:reduct} is proved.

\subsection{Proof of Theorem \ref{thm:hard}:}\label{sec:analysis}


If $\mathcal{C}$ is a recursively enumerable class, and $\mathcal{C}\nsubseteq\mathcal{C}_\lambda\cup \mathcal{C}^*$ for any constant $\lambda$, then for every $k\ge2$, let $H_k$ be the first graph in $\mathcal{C}$ where $H_k$ is not a star and has at least $2k^{10}$ edges. The time for finding $H_k$ is $f(k)$ for some function $f$. From Lemma \ref{lem:hk} we know that $H_k\in\mathcal{H}_k$, so that we can use Theorem \ref{thm:reduct} to construct the $\textsc{SLSN}_\mathcal{C}$ instance with demand $H_k$.

The parameter $p=|H_k|$ of the instance is a function just of $k$, and the construction time is \textsc{FPT} from Theorem \ref{thm:reduct}. Therefore this is a \textsc{FPT} reduction from the \textsc{MCC} problem to the unit-length unit-cost $\textsc{SLSN}_\mathcal{C}$ problem. Thus Theorem \ref{thm:mcc} implies that the unit-length unit-cost $\textsc{SLSN}_\mathcal{C}$ problem is \textsc{W}$[1]$-hard with parameter $p$.
\qed

\section{Algorithms for Arbitrary-Length Arbitrary-Cost \textsc{SLSN}}\label{sec:approxAlgorithm}

The idea for the algorithms for arbitrary-length arbitrary-cost \textsc{SLSN} is the same as that for unit-length arbitrary cost. However, the arbitrary-lengths increase the difficulty of the problem. For example, we cannot use Bellman-Ford algorithm to find the lowest cost path within a certain distance bound. We also cannot go over all the possible lengths of the paths in the dynamic programming algorithm, because it will take exponential time.

In order to recover from this, we utilize some techniques in the $(1+\varepsilon)$-approximation algorithm for the \textsc{Restricted Shortest Path} problem, where the problem is the special case of \textsc{SLSN} with $p=1$.

\subsection{Preliminaries}\label{app:pre}

In this section, we will introduce two algorithms: the first one gives a bound on the optimal solution of the \textsc{SLSN} problem, and the second one finds the shortest path under some flexible cost constraint. The algorithms are similar to the algorithms for the \textsc{Restricted Shortest Path} problem.

Given an \textsc{SLSN} instance, the first algorithm $OptLow$ orders all the edges in $E$ by cost and starts from the lowest one. In each iteration $i$, let $e_i$ be the edge considered, and let $G_{e_i}$ be the graph that contains all edges in $E$ with cost at most $c(e_i)$. The algorithm checks if $G_{e_i}$ contains a feasible solution of the \textsc{SLSN} instance, and returns $C=c(e_i)$ if a solution exists. The pseudocode is in Algorithm \ref{alg:opt}.

\begin{algorithm}
\caption{$OptLow(G=(V,E),c,l,H)$}\label{alg:opt}
\begin{algorithmic}
\State Order all the edges in $E$ by the cost and get $c(e_1)\le c(e_2)\le\mathellipsis\le c(e_{|E|})$
\For{$i=1,\mathellipsis,|E|$}
	\State $G_{e_i}\gets$ the graph that contains all edges in $E$ with cost at most $c(e_i)$
	\If{$G_{e_i}$ has a feasible solution}
		\State $C\gets c(e_i)$
		\State {\bf break}
	\EndIf 
\EndFor
\\\Return $C$
\end{algorithmic}
\end{algorithm}

\begin{lemma}\label{lem:opt}
Let $C$ be the solution returned by Algorithm \ref{alg:opt} and $OPT$ be the cost of the optimal solution for the \textsc{SLSN} instance $(G=(V,E),c,l,H)$. Then $C\le OPT\le n^2C$, and the running time is polynomial in $n$.
\end{lemma}
\begin{proof}
Since a graph in which all edges have cost less than $C$ does not have a feasible solution, we know that every feasible solution contains an edge which has cost at least $C$, so the optimal solution has cost at least $C$.

For the upper bound, because graph $G_{e_i}$ contains a feasible subgraph, so there is a feasible solution with cost at most $n^2C$, thus the optimal solution has cost at most $n^2C$.

Because we can use a standard shortest path algorithm for each pair of demands to test the feasibility, we can see that the running time is polynomial in $n$.
\end{proof}

Following is the second algorithm, which is aiming to find a low cost path under certain distance bound. Algorithm \ref{alg:minDist} is essentially a dynamic programming algorithm on graph $G$, which first defined new costs $c_e^*=\left\lceil\frac{n\cdot c(e)}{\varepsilon C}\right\rceil$ for all $e\in E$, and then calculate the shortest path from $s$ to $t$ with bounded new cost $\left\lfloor\frac{n}{\varepsilon}\right\rfloor$.

\begin{algorithm}
\caption{$MinDist(G=(V,E),c,l,s,t,\varepsilon,C)$}\label{alg:minDist}
\begin{algorithmic}
\For{$e\in E$}
	\State $c_e^*\gets\left\lceil\frac{n\cdot c(e)}{\varepsilon C}\right\rceil$
\EndFor
\For{$v\in V\setminus\{s\}$}
	\State $d(v,0)\gets\infty$
\EndFor
\State $d(s,0)\gets0$
\For{$i=1,\mathellipsis,\left\lfloor\frac{n}{\varepsilon}\right\rfloor$}
	\For{$v\in V$}
		\State $d(v,i)\gets\min_{e=(u,v)\in E,c_e^*\le i}l(e)+d(u,i-c_e^*)$
	\EndFor
\EndFor
\State $C^*\gets\arg\min_{i\in[\left\lfloor\frac{n}{\varepsilon}\right\rfloor]}d(t,i)$
\If{$d(t,C^*)<\infty$}
\\\quad\ \ \Return the corresponding path for $d(t,C^*)$
\Else
\\\quad\ \ \Return $\varnothing$
\EndIf
\end{algorithmic}
\end{algorithm}

\begin{lemma}\label{lem:minDist}
Given a graph $G=(V,E)$, a cost function $c$, a length function $l$, a pair of vertices $(s,t)$, a constant $\varepsilon>0$, and a cost bound $C$, if there exists a path between $s$ and $t$ with cost at most $(1-2\varepsilon)C$ and distance $D$, then Algorithm \ref{alg:minDist} returns a path between $s$ and $t$ with cost at most $C$ and distance at most $D$. The running time is polynomial in $\frac{n}{\varepsilon}$.
\end{lemma}
\begin{proof}
The running time of this algorithm is clearly polynomial in $\frac{n}{\varepsilon}$ because there is only $poly(\frac{n}{\varepsilon})$ slots of $d(v,i)$.

We first claim that $d(v,i)$ is the minimal length of a path from $s$ to $v$ under new cost $i$. This can be proven easily by induction. The base case is $i=0$, where $d(s,0)=0$, and for other $v\in V\setminus\{s\}$, the minimal length is infinity. For the inductive step, we consider the last edge $\{u,v\}$ of the shortest path from $s$ to $v$ under new cost $i$. By removing this edge from the path, it must be a path from $s$ to $u$ under new cost $i-c_{\{u,v\}}^*$. Therefore our claim holds.

If there exists a path between $s$ and $t$ with cost at most $(1-2\varepsilon)C$ and distance $D$, let the edge set of this path be $S$. Then because $S$ has at most $n$ edges, we know that the new total cost of $S$ is
\[\sum_{e\in S}c_e^*=\sum_{e\in S}\left\lceil\frac{n\cdot c(e)}{\varepsilon C}\right\rceil\le\frac{n\cdot\sum_{e\in S}c(e)}{\varepsilon C}+n\le\frac{n(1-2\varepsilon)C}{\varepsilon C}+n\le\left\lfloor\frac{n}{\varepsilon}\right\rfloor.\]
Thus the algorithm must return a non-empty set $S'$, which connects $s$ and $t$ with distance at most $D$.

Now we can calculate the original cost of path $S'$. We know that
\[\left\lfloor\frac{n}{\varepsilon}\right\rfloor\ge\sum_{e\in S'}c_e^*=\sum_{e\in S'}\left\lceil\frac{n\cdot c(e)}{\varepsilon C}\right\rceil>\frac{n\cdot\sum_{e\in S'}c(e)}{\varepsilon C}.\]
Thus
\[\sum_{e\in S'}c(e)<\left\lfloor\frac{n}{\varepsilon}\right\rfloor\cdot\frac{\varepsilon C}{n}\le\frac{n}{\varepsilon}\cdot\frac{\varepsilon C}{n}=C.\]

Therefore, the algorithm returns a path $S'$, which has cost at most $C$ and distance at most $D$.
\end{proof}

\subsection{Constant Number of Demands}\label{sec:approxConst}

With the algorithms in Section \ref{app:pre}, we can now introduce our \textsc{FPTAS} algorithm for arbitrary-length arbitrary-cost $\textsc{SLSN}_{\mathcal{C}_\lambda}$. Since Lemma \ref{lem:intersect} still holds for the arbitrary length case, we can use the same idea as in Algorithm \ref{alg:const}. However, in the arbitrary length case we can not guess the exact length of each subpath, because there are too many possible lengths. We even can not guess an approximate length for each subpath, because any violation on the length bound may make the solution infeasible. Therefore, we switch to guess the approximate cost of each subpath to solve this problem.

The algorithm for arbitrary-length arbitrary-cost $\textsc{SLSN}_{\mathcal{C}_\lambda}$ first bound the optimal cost using Algorithm \ref{alg:opt}. Then guess the endpoint set $Q$ and the set $E'$ which intuitively represents how the vertices in $Q\cup\bigcup_{i=1}^p\{s_i,t_i\}$ connected to each other, the same way as in Algorithm \ref{alg:const}. After that, the algorithm switch to guess the cost $c'$ of each subpath, basically up to a $(1+\varepsilon)$ error. Finally, the algorithm connect each pair of $u,v\in V$ where $\{u,v\}\in E'$ by shortest paths with restricted cost $c'(\{u,v\})$ using Algorithm \ref{alg:minDist}, check the feasibility, and output the optimal solution. The detailed algorithm is in Algorithm \ref{alg:approxConst}.

\begin{algorithm}
\caption{arbitrary-length arbitrary-cost $\textsc{SLSN}_{\mathcal{C}_\lambda}$}\label{alg:approxConst}
\begin{algorithmic}
\State $C\gets OptLow(G=(V,E),c,l,H)$
\State $M\gets\sum_{e\in E}c(e)$
\State $S\gets E$
\For{$Q\subseteq V$ where $|Q|\le p(p-1)$}
	\State $Q'\gets Q\cup\bigcup_{i=1}^p\{s_i,t_i\}$
	\For{$E'\subseteq\{\{u,v\}\mid u,v\in Q',u\ne v\}$, $c':E'\rightarrow\{-\left\lceil2\log_{1+\varepsilon}n\right\rceil,\mathellipsis,-1,0,1,\mathellipsis,\left\lceil2\log_{1+\varepsilon}n\right\rceil+3\}$}
		\State $T\gets\varnothing$
		\For{$\{u,v\}\in E'$}
			\State $T\gets T\cup MinDist(G,c,l,u,v,\varepsilon,(1+\varepsilon)^{c'(\{u,v\})}C)$
		\EndFor
		\If{$T$ is a feasible solution and $\sum_{e\in T}c(e)<M$}
			\State $M\gets\sum_{e\in T}c(e)$
			\State $S\gets T$
		\EndIf
	\EndFor
\EndFor
\\
\Return $S$
\end{algorithmic}
\end{algorithm}

\begin{claim}\label{claim:approxConstTime}
The running time for Algorithm \ref{alg:approxConst} is is $(\frac{n}{\varepsilon})^{O(p^4)}$.
\end{claim}
\begin{proof}
We know that $OptLow$ can be done in polynomial time by Lemma \ref{lem:opt}. Similar to Algorithm \ref{alg:const}, we can also see that there are at most $n^{p(p-1)}$ possible $Q$, and for each $Q$ there are at most $2^{(p(p-1)+2p)^2}$ possible $E'$, and at most $(2\cdot\left\lceil2\log_{1+\varepsilon}n\right\rceil+3)^{(p(p-1)+2p)^2}$ possible $c'$. The algorithm \ref{alg:minDist} in the inner loop takes $poly(\frac{n}{\varepsilon})$ running time. Thus the total running time is at most $n^{p(p-1)}\cdot 2^{(p(p+1))^2}\cdot(\frac{5\log n}{\varepsilon})^{(p(p+1))^2}\cdot poly(\frac{n}{\varepsilon})+poly(n)$.
\end{proof}

\subsubsection{Proof of Theorem \ref{thm:approxConst}:}

The running time has been proven in Claim \ref{claim:approxConstTime}, which is polynomial in $\frac{n}{\varepsilon}$ if $\lambda\ge p$ is a constant. The correctness is also similar to Algorithm \ref{alg:const}. Because the algorithm only returns feasible solution, so we only need to show that this algorithm returns a solution with cost at most the cost of the optimal solution.

We define $S^*$, $P_i^*$, $Q^*$, $Q'^*$, $l'^*$, $P_{\{u,v\}}^*$ the same as in the proof of Theorem \ref{thm:const}. We further define $c'^*(\{u,v\})$ as the cost of $P_{\{u,v\}}^*$ for each $\{u,v\}\in E'$.

Since the algorithm iterates over all possibilities for $Q$, $E'$ and $c'$, there is some iteration in which $Q=Q'^*$, $E'=E'^*$, and  $c'\equiv\max\left\{\left\lceil\log_{1+\varepsilon}\frac{c'^*}{(1-2\varepsilon)C}\right\rceil,-\left\lceil2\log_{1+\varepsilon}\frac{n^2}{\varepsilon}\right\rceil\right\}$. The reason that this $c'$ must have been iterated is because of Lemma \ref{lem:opt}. We can see that, $c'^*(\{u,v\})\le OPT\le n^2C$ for every $\{u,v\}\in E'^*$, and so that
\[\left\lceil\log_{1+\varepsilon}\frac{c'^*(\{u,v\})}{(1-2\varepsilon)C}\right\rceil\le\left\lceil\log_{1+\varepsilon}\frac{n^2C}{(1-2\varepsilon)C}\right\rceil\le\left\lceil2\log_{1+\varepsilon}n\right\rceil+3.\]
We will show that the algorithm also must find a feasible solution in this iteration.

For each $i\in[p]$, the path $P_i^*$ is partitioned to edge-disjoint subpaths by $Q'^*$. Let $q_i$ be the number of subpaths, and let the endpoints be $s_i=v_{i,0},v_{i,1},\mathellipsis,v_{i,q_i-1},v_{i,q_i}=t_i$. We further let these subpaths be $P_{\{s_i,v_{i,1}\}}^*,P_{\{v_{i,1},v_{i,2}\}}^*,\mathellipsis,P_{\{v_{i,q_i-1},t_i\}}^*$. By the definition of $l'^*$ and $c'^*$, for each $j\in[q_i]$, there must be a path between $v_{i,j-1}$ and $v_{i,j}$ with length at most $l'^*(\{v_{i,j-1},v_{i,j}\})$ and cost at most
\[c'^*(\{v_{i,j-1},v_{i,j}\})\le(1-2\varepsilon)\cdot(1+\varepsilon)^{\{c'(v_{i,j-1},v_{i,j}\})}C.\]

Therefore, by Lemma \ref{lem:minDist} we know that the edge set $T$ in this iteration must contains a path between $u$ and $v$ with length at most $l'^*(\{v_{i,j-1},v_{i,j}\})$ and cost at most
\[(1+\varepsilon)^{c'(\{v_{i,j-1},v_{i,j}\})}C\le\max\left\{\frac{c'^*(\{v_{i,j-1},v_{i,j}\})}{1-2\varepsilon},\frac{\varepsilon C}{n^2}\right\}.\]
Because the summation $\sum_{j=1}^{q_i}l'^*(\{v_{i,j-1},v_{i,j}\})$ is at most $L$, we know that the edge set $T$ in this iteration must satisfies demand $\{s_i,t_i\}$ for each $i\in[p]$. Therefore it is a feasible solution.

We can also see that the cost of the edge set $T$ in this iteration is at most
\begin{align}
\sum_{\{u,v\}\in E'^*}\max\left\{\frac{c'^*(\{u,v\})}{1-2\varepsilon},\frac{\varepsilon C}{n^2}\right\}&\le\frac{1}{1-2\varepsilon}\sum_{e\in E'^*}c'^*(\{u,v\})+\varepsilon C\nonumber\\
&\le\frac{OPT}{1-2\varepsilon}+\varepsilon C\label{eqn:c}\\
&\le\frac{OPT}{1-2\varepsilon}+\varepsilon OPT\label{eqn:C}\\
&\le(1+4\varepsilon)OPT.\nonumber
\end{align}

Equation (\ref{eqn:c}) is because the all the $P_{u,v}^*$ are edge-disjoint, and thus we have $OPT=\sum_{(u,v)\in E'^*}c'^*(\{u,v\})$. Equation (\ref{eqn:C}) is because we know that $OPT\ge C$ by Lemma \ref{lem:opt}.

Therefore, the algorithm outputs a $(1+4\varepsilon)$-approximation of the optimal solution. By replacing $\varepsilon$ with $\frac{\varepsilon}{4}$ in the whole algorithm, we get a $(1+\varepsilon)$-approximation.
\qed

\subsection{Star Demand Graphs ($\textsc{SLSN}_{\mathcal{C}^*}$)}\label{sec:approxStar}

For the case that the demand graph is a star, let $s=s_1=s_2=\mathellipsis=s_p$, and $T=\{t_1,\mathellipsis,t_p\}$.

We first bound the optimal cost using Algorithm \ref{alg:opt}, and then assign a new cost for each edge depending our bound of the optimal cost. Finally, we use a dynamic programming algorithm which is similar to the algorithm for \textsc{DST} to solve the problem under the new edge costs, and we can show that it is a $(1+\varepsilon)$-approximation to the optimal solution in the original edge cost. The detailed algorithm is in Algorithm \ref{alg:approxStar}.

Here we are aiming to set $d(v,R,j)$ as the smallest height of a tree, such that the root is $v$, the total new cost is at most $j$, and it contains all the vertices in $R$. Then, we can find the minimal $j$ which makes $d(v,T,j)\le L$, and this $j$ is the minimal cost of a feasible solution under the new cost. Note that we only need to consider the height of trees because the optimal solution in this case is always a tree.

\begin{algorithm}
\caption{arbitrary-length arbitrary-cost $\textsc{SLSN}_{\mathcal{C}^*}$}\label{alg:approxStar}
\begin{algorithmic}
\State $C\gets OptLow(G=(V,E),c,l,H)$
\For{$e\in E$}
	\State $c_e^*\gets\left\lceil\frac{n\cdot c(e)}{\varepsilon C}\right\rceil$
\EndFor
\For{$v\in V,R\subseteq T$, and $j\in[\left\lceil\frac{n^3(1+\varepsilon)}{\varepsilon}\right\rceil]$}
	\State $d(v,R,j)\gets\begin{cases}0,&\mbox{if }|R|=1\mbox{ and }v\in R,\mbox{ or }R=\varnothing\\\infty,&\mbox{otherwise}\end{cases}$
\EndFor
\For{$j=1,\mathellipsis,\left\lceil\frac{n^3(1+\varepsilon)}{\varepsilon}\right\rceil$}
	\For{$i=1,\mathellipsis,p$}
		\For{$v\in V,R\subseteq T$ with $|R|=i$}
			\State $d(v,R,j)\gets\min_{v'\in V,R'\subseteq R,k\le j-c_{\{v,v'\}}^*}(l(\{v,v'\})+\max\{d(v',R'\setminus\{v\},k),d(v',R\setminus R'\setminus\{v\},j-c_{\{v,v'\}}^*-k)\})$
		\EndFor
	\EndFor
\EndFor
\For{$j=1,\mathellipsis,\left\lceil\frac{n^3(1+\varepsilon)}{\varepsilon}\right\rceil$}
	\If{$d(s,T,j)\le L$}
		\\\quad\quad\quad\Return the corresponding tree for $d(s,T,j)$
	\EndIf
\EndFor
\end{algorithmic}
\end{algorithm}

\begin{lemma}\label{lem:tree}
The optimal solution $S^*$ of the $\textsc{SLSN}_{\mathcal{C}^*}$ problem is always a tree.
\end{lemma}
\begin{proof}
We can assign path $P_1,\mathellipsis P_p$ as in the Lemma \ref{lem:intersect}. Because $S^*$ is an optimal solution, it will not contain any edge other than the edges in $P_1,\mathellipsis P_p$. If there is a cycle in $S^*$, then there are two paths $P_i$ and $P_j$ intersect at a vertex $v$ other than the root $s$, and the paths to $v$ are different, which contradict with Lemma \ref{lem:intersect}. Therefore $S^*$ is always a tree.
\end{proof}

We again first prove the running time.

\begin{claim}\label{claim:approxStarTime}
Algorithm \ref{alg:approxStar} runs in time $O(4^p\cdot poly(\frac{n}{\varepsilon}))$.
\end{claim}
\begin{proof}
Because running algorithm $OptLow$ and setting new costs runs in polynomial time, we only need to prove the time of the dynamic programming part. We can see that, $d$ has at most $n\cdot 2^p\cdot\left\lceil\frac{n^3(1+\varepsilon)}{\varepsilon}\right\rceil$ slots, and filling each of them takes at most $n\cdot 2^p$ time, so the total running time is $O(4^p\cdot poly(\frac{n}{\varepsilon}))$.
\end{proof}

We then prove the correctness of the dynamic programming part.

\begin{lemma}\label{lem:dp}
For each $v\in V$, $S\subseteq T$, and $j\in[\left\lceil\frac{n^3(1+\varepsilon)}{\varepsilon}\right\rceil]$, the $d(v,R,j)$ stores the smallest height of a tree, such that the root is $v$, the total new cost is at most $j$, and it contains all the vertices in $R$.
\end{lemma}
\begin{proof}
We prove the lemma using induction. The base case is that $R=\varnothing$ or $R=\{v\}$, which has already been initialized.

The algorithm fill all the $d(v,R,j)$ with the ascending order of $j$ and then ascending order of $|R|$, so for any $v'\in V$, $R'\subseteq R$, and $k\le j-c_{\{v,v'\}}^*$, we know that $d(v',R'\setminus\{v\},k)$ and $d(v',R\setminus R'\setminus\{v\},j-c_{\{v,v'\}}^*-k)$ must have already been filled before filling $d(v,R,j)$. We will show that when updated, $d(v,R,j)$ is at most and at least the smallest height of a tree, such that the root is $v$, the total new cost is at most $j$, and it contains all the vertices in $R$.

For the ``at most'' part, let $S$ be the lowest tree, such that the root is $v$, the total new cost is at most $j$, and it contains all the vertices in $R$. If $S$ is not in the base case, then either $v$ has degree $1$, or $v$ has degree more than $1$ in $S$.

If $v$ has degree $1$, then there must be a $v'$ which is adjacent to $v$, and the tree rooted at $v'$ is the lowest height tree, which contains all the vertices in $R\setminus\{v\}$, and the total cost is at most $j-c_{\{v,v'\}}^*$. This case is already considered in the algorithm by setting $R'=R$ and $k=j-c_{\{v,v'\}}^*$, thus in this case $d(v,R,j)$ is at most the height of $S$.

If $v$ has degree more than $1$, then $S$ can be split to two trees $S_1$ and $S_2$ with the same root $v$. Let $R_1=R\cap S_1$, then the height of $S$ is at least the lowest possible height of $S_1$, and also at least the lowest possible height of $S_2$. This case is considered in the algorithm by setting $v'=v$, $R'=R_1$ and $k$ be the cost of $S_1$, thus in this case $d(v,R,j)$ is also at most the height of $S$.

For the ``at least'' part, we only need to show that there exist a tree with height $d(v,R,j)$ such that the root is $v$, the total new cost is at most $j$, and it contains all the vertices in $R$. Let $v^*$, $R^*$, and $k^*$ be the value of $v'$, $R$, and $k$ which gives the minimum value of $d(v,R,j)$. Then after removed redundant edges, the union of the edge $\{u,v\}$, the tree for $d(v^*,R^*\setminus\{v\},k^*)$, and the tree for $d(v^*,R\setminus R^*\setminus\{v\},j-c_{\{v,v^*\}}^*-k^*)$ is a tree which the root is $v$, the total new cost is at most $j$, and it contains all the vertices in $S$, with height $d(v,R,j)$.

Therefore $d(v,R,j)$ is correctly set to what we want.
\end{proof}


No we can finally prove our Theorem \ref{thm:approxStar}.

\subsubsection{Proof of Theorem \ref{thm:approxStar}}

The running time has already been proven in Claim \ref{claim:approxStarTime}. Now we prove the correctness.

Let $S^*$ be the optimal solution and let $OPT=\sum_{e\in S^*}c(e)$. Then, the new cost of this solution is at most
\begin{align*}
\sum_{e\in S^*}c_e^*&=\sum_{e\in S^*}\left\lceil\frac{n\cdot c(e)}{\varepsilon C}\right\rceil\\
&\le\frac{n\cdot\sum_{e\in S^*}c(e)}{\varepsilon C}+n\\
&\le\frac{n\cdot OPT}{\varepsilon C}+\frac{n\cdot\varepsilon OPT}{\varepsilon C}\\
&=\frac{n}{\varepsilon C}(1+\epsilon)OPT\\
&\le\left\lceil\frac{n^3(1+\varepsilon)}{\varepsilon}\right\rceil,
\end{align*}
because from Lemma \ref{lem:tree} we know that $|S^*|\le n$ and from Lemma \ref{lem:opt} we know that $C\le OPT\le n^2C$.

Since $S^*$ is a tree with height at most $L$, such that the root is $s$, the total new cost is at most $\frac{n}{\varepsilon C}(1+\epsilon)OPT$, and it contains all the vertices in $T$, from Lemma \ref{lem:dp} and the last section of Algorithm \ref{alg:approxStar} we know that the algorithm must returns a tree $S$ with height at most $L$, the total new cost is at most $\frac{n}{\varepsilon C}(1+\epsilon)OPT$, and it contains all the vertices in $T$, which is a feasible solution.

Now we calculate the original cost of $S$. Because
\[\frac{n}{\varepsilon C}(1+\epsilon)OPT\ge\sum_{e\in S}c_e^*=\sum_{e\in S}\left\lceil\frac{n\cdot c(e)}{\varepsilon C}\right\rceil\ge\frac{n\cdot\sum_{e\in S}c(e)}{\varepsilon C},\]
we know that $\sum_{e\in S}c(e)\le(1+\epsilon)OPT$, which is a $(1+\varepsilon)$ approximation to the optimal solution.
\qed

\section{Hardness for Unit-Length Polynomial-Cost \textsc{SLSN}}\label{sec:approxHard}

\subsection{Preliminaries}\label{sec:approxPre}

In this section, we will do a \textsc{FPT} reduction from the \textsc{Multi-Colored Densest} $k$\textsc{-Subgraph} (\textsc{Multi-Colored D}$k$\textsc{S}) problem to the unit-length polynomial-cost $\textsc{SLSN}_\mathcal{C}$ problem with a $\mathcal{C}\nsubseteq\mathcal{C}_\lambda\cup \mathcal{C}^*$. Here is the definition of the \textsc{Multi-Colored D}$k$\textsc{S} problem.

\begin{definition}[\textsc{Multi-Colored Densest} $k$\textsc{-Subgraph}]
Given a graph $G=(V,E)$, a number $k\in\mathbb{N}$, a coloring function $c:V\rightarrow[k]$, and a factor $\alpha<1$. The objective of the \textsc{Multi-Colored D}$k$\textsc{S} problem is to distinguish the following two cases:
\begin{itemize}
\item There is a $k$-clique in $G$, where each vertex has different color.
\item Every subgraph of $G$ induced by $k$ vertices contains less than $\alpha\cdot\binom{k}{2}$ edges.
\end{itemize}
\end{definition}

Previously, it has been proven that, assume Gap-ETH holds, then there is no \textsc{FPT} algorithm for \textsc{Multi-Colored D}$k$\textsc{S} even with $\alpha=o(1)$. Formally, the theorem is as follows.

\begin{theorem}[\cite{chitnis2017parameterized}, Corollary 24]\label{thm:dks}
Assuming (randomized) Gap-ETH, for any function $h(k)=o(1)$ and any function $f$, there is no $f(k)\cdot n^{O(1)}$-time algorithm that solves \textsc{Multi-Colored D}$k$\textsc{S} with factor $\alpha=k^{-h(k)}$.
\end{theorem}

We can easily get a weaker version of this theorem which $\alpha=O(1)$.

\begin{corollary}\label{col:dks}
For any constant $0 < \alpha < 1$, for any function $f$, assuming (randomized) Gap-ETH there is no $f(k)\cdot n^{O(1)}$-time algorithm that solves \textsc{Multi-Colored D}$k$\textsc{S} with factor $\alpha$.
\end{corollary}
\begin{proof}
We can set $h(k)=\log_k\frac{1}{\alpha}$ in Theorem \ref{thm:dks}.
\end{proof}

\subsection{Reduction}\label{sec:approxReduct}

\begin{theorem}\label{thm:approxReduct}
Let $1\ge\epsilon > 0$ be an arbitrary constant, and let $G=(V,E)$, coloring function $c : V \rightarrow [k]$, and factor $\varepsilon$ be a \textsc{Multi-Colored D}$k$\textsc{S} instance.  Let $H \in \mathcal H_k$.  Then we can construct a unit-length polynomial-cost \textsc{SLSN} instance $(G',L)$ with demand graph $H$ in $poly(|V||H|)$ time, and there exists a function $g$ (computable in time $poly(|H|)$) such that
\begin{itemize}
\item If there is a $k$-clique in $G$, where each vertex has different color, then the \textsc{SLSN} instance has a solution with cost $g(H)$.
\item If every subgraph of $G$ induced by $k$ vertices contains less than $\varepsilon\cdot\binom{k}{2}$ edges, then the optimal cost of the \textsc{SLSN} instance is at least $\left(\frac{5}{4}-\varepsilon\right)g(H)$.
\end{itemize}
\end{theorem}

As in the unit-length unit-cost setting, we will first design a reduction for demand graphs $H\in\{H_{k,0}^*,H_{k,1}^*,H_{k,2}^*,H_{k,k}\}\cup\mathcal{H}_{2,k}$ first, and then consider the general $H \in \mathcal{H}_k$.

\subsubsection{Case \ref{case:disjoint}: $H_{k,0}^*$}

Let $G=(V,E)$ with coloring function $c : V \rightarrow [k]$ and factor $\varepsilon$ be a \textsc{Multi-Colored D}$k$\textsc{S} instance. We create a unit-length and polynomial-cost \textsc{SLSN} instance $G'$ with demand graph $H_{k,0}^*$ as following.

We again use the length-weighted graph $G_k^*$ constructed in Section \ref{sec:disjoint}. We change the cost of edges in $E_2\cup E_4$ to $4k^4$, while keeping the cost equal to the length for the rest of the edges.

$G'$ is again a graph that each edge $e\in E_k^*$ is replaced by a $length(e)$-hop path. Where the cost of edges is divided equally for each hop. The demands are the same as the demands in Section \ref{sec:disjoint}, and $L$ is still $4k^2$. The construction still takes $|V||H_{k,0}^*|$. The function $g$ is slightly different, where $g(H_{k,0}^*)=6k^6-6k^5+3k^4+k$, this function is also computable in $poly(H_{k,0}^*)$ time.

Using the same solution as in the proof of Lemma \ref{lem:disjointCost}, we can see that, If there is a multi-colored clique of size $k$, then the \textsc{SLSN} instance has a solution with cost
\[4k^4-4k^3+\frac{3}{2}k^2+\frac{5}{2}k+\left(\binom{k}{2}+k(k-1)\right)\cdot(4k^4-1)=6k^6-6k^5+3k^4+k,\] because the cost of $\binom{k}{2}+k(k-1)$ edges in this solution is changed from $1$ to $4k^4$.

The other direction for the correctness is the following lemma.

\begin{lemma}
Let $S$ be an optimal solution for the \textsc{SLSN} instance $(G',L)$ with demand graph $H_{k,0}^*$. If $S$ has cost at most $\left(\frac{5}{4}-\frac{\varepsilon}{4}\right)g(H_{k,0}^*)$, then there is a subgraph of $G$ with $\varepsilon\cdot\binom{k}{2}$ edges.
\end{lemma}
\begin{proof}
For each $i,j\in[k]$ where $i\ne j$, let $P_{i,j}$ be a (arbitrarily chosen) path in $S$ which connects $r$ and $l_{i,j}$ with length at most $L=4k^2$. Let $\mathcal{P}=\{P_{i,j}\mid i,j\in[k],i\ne j\}$ be the set of all these paths. We also let $P_y$ be a (arbitrarily chosen) path in $S$ which connects $y_0$ and $y_k$ with length at most $L$.

From Claim \ref{claim:singlePath}, $P_y$ can be divided to $k$ subpaths, each correlates to a vertex $v_i$. We will show that the induced subgraph on vertex set $\{v_1,\mathellipsis,v_k\}$ has at least $\varepsilon\cdot\binom{k}{2}$ edges. In fact, let $R=\{\{v_i,v_j\}\mid i,j\in[k],i\ne j, P_{i,j}\cap E_2=P_{j,i}\cap E_2,P_y\cap P_{i,j}\cap E_4\ne\varnothing,P_y\cap P_{j,i}\cap E_4\ne\varnothing\}$, we will show that $R\subseteq E$ and $|R|\ge\varepsilon\cdot\binom{k}{2}$.

For any $\{v_i,v_j\}\in R$, by looking at the definition of $G_k^*$ and the form of the path $P_y$ and $P_{i,j}$ in Claim \ref{claim:leafPath} and \ref{claim:singlePath}, we know that if $P_y\cap P_{i,j}\cap E_4$ is not an empty set, then it must contain only one edge $\{x_{v_{i},j},x_{v_{i},j}'\}$. Similarly, $P_y\cap P_{j,i}\cap E_4$ must be the edge $\{x_{v_{j},i},x_{v_{j},i}'\}$. From this, we can see that $P_{i,j}\cap E_2=P_{j,i}\cap E_2$ must be the edge $\{z_{\{i,j\}},z_{\{v_i,v_j\}}\}$, because $z_{\{v_i,v_j\}}$ is the only vertex which is adjacent to both $x_{v_{i},j}$ and $x_{v_{j},i}$. Therefore by the definition of $E_2$, we know that $\{v_i,v_j\}$ is an edge of $E$, which means $R\subseteq E$. Thus the only thing left is to show that $|R|\ge\varepsilon\cdot\binom{k}{2}$.

From Claim \ref{claim:singlePath}, because $P_y$ contains $k$ subpaths, and each subpath contains $k-1$ edges in $E_4$, we know that $|P_y\cap E_4|\ge k(k-1)$. From Claim \ref{claim:leafPath}, because for each $i,j\in[k]$ where $i\ne j$, there is at least one edge in $P_{i,j}\cap E_4$, and they must be different from each other, we know that $\left|\bigcup_{i,j\in[k],i\ne j}P_{i,j}\cap E_4\right|\ge k(k-1)$. Let $x=\left|S\cap E_4\right|=\left|(P_y\cup\bigcup_{i,j\in[k],i\ne j}P_{i,j})\cap E_4\right|$, then
\[\left|P_y\cap\bigcup_{i,j\in[k],i\ne j}P_{i,j}\cap E_4\right|=\left|P_y\cap E_4\right|+\left|\bigcup_{i,j\in[k],i\ne j}P_{i,j}\cap E_4\right|-x\ge2k(k-1)-x.\]
Let $T=\{P_{i,j}\mid i,j\in[k],i\ne j, P_y\cap P_{i,j}\cap E_4\ne\varnothing\}$, then $|T|\ge2k(k-1)-x$, because each $P_{i,j}$ can share at most one edge with $P_y$.

We also know that each edge $\{z_{\{i,j\}},z_{e}\}\in S\cap E_2$ can only appear in at most two different paths, which are $P_{i,j}$ and $P_{j,i}$. Let $y=|S\cap E_2|$. From Claim \ref{claim:leafPath} we know that each path in $T$ must contain at least one edge in $S\cap E_2$, thus there are at least $|T|-y$ edges in $S\cap E_2$ which appear in two different paths in $T$. Therefore $|R|\ge|T|-y$.

Now we calculate the size of $R$. Because every edge in $E_2\cup E_4$ has cost $4k^4$, and the total cost is less than $\left(\frac{5}{4}-\frac{\varepsilon}{4}\right)g(H_{k,0}^*)$, thus we have
\[x+y=|S\cap E_2|+|S\cap E_4|\le\left\lfloor\frac{\left(\frac{5}{4}-\frac{\varepsilon}{4}\right)g(H_{k,0}^*)}{4k^4}\right\rfloor\le\left(\frac{15}{8}-\frac{3\varepsilon}{8}\right)k(k-1),\]
so that
\[|R|\ge|T|-y\ge2k(k-1)-x-y\ge2k(k-1)-\left(\frac{15}{8}-\frac{3\varepsilon}{8}\right)k(k-1)\ge\varepsilon\cdot\binom{k}{2}.\]

Therefore the induced subgraph with vertex set $\{v_1,\mathellipsis,v_k\}$ has at least $\varepsilon\cdot\binom{k}{2}$ edges.
\end{proof}

\subsubsection{Case \ref{case:one}, \ref{case:two}, and \ref{case:matching}:}

In this setting, the construction of Case \ref{case:one}, \ref{case:two} still keep the same as Case \ref{case:disjoint}, and the change for \ref{case:matching} are basically the same as the unit-cost setting.

\text{}\\
\textbf{Case \ref{case:one}}: $H_{k,1}^*$

We use the same $G_k^*$, $G'$ and $L$ in the construction of the \textsc{SLSN} instance for demand graph $H_{k,0}^*$, and also set $g(H_{k,1}^*)=6k^6-6k^5+4k^4+k$. The only difference is the demand graph. Besides the demand of $\{r,l_{i,j}\}$ for all $i,j\in[k]$ where $i\ne j$, and $\{y_0,y_k\}$, there is a new demand $\{r,y_0\}$. Clearly this new demand graph is a star with $(k(k-1)+1)$ leaves, and an edge in which
exactly one of the endpoints is a leaf of the star, so it is isomorphic to $H_{k,1}^*$.

Assume there is a multi-colored clique of size $k$ in $G$. The paths connecting previous demands in the solution of the \textsc{SLSN} instance are the same as Case \ref{case:disjoint}. The path between $r$ and $y_0$ is $r$ -- $z_{\{1,2\}}$ -- $z_{\{v_1,v_2\}}$ -- $x_{v_1,2}$ -- $y_0$. All the edges in this path is already in the previous paths, so the cost remains the same. The length of this path is $2+1+2k^2-2+4=2k^2+5<4k^2$, which satisfies the length bound.

Assume there is a solution for the \textsc{SLSN} instance $(G',L,H_{k,1}^*)$ with total cost less than $\left(\frac{5}{4}-\frac{\varepsilon}{4}\right)g(H_{k,1}^*)$. The proof of existing a subgraph of $G$ with $\varepsilon\cdot\binom{k}{2}$ edges is the same as Case \ref{case:disjoint}.

\text{}\\
\textbf{Case \ref{case:two}}: $H_{k,2}^*$

As in Case \ref{case:one}, only the demand graph changes. The new demand graph is the same as in Case \ref{case:one} but again with a new demand $\{r,y_k\}$. Since $\{r,y_0\}$was already a demand, our new demand graph is a star with $(k(k-1)+2)$ leaves (the $l_{i,j}$'s and $y_0$ and $y_k$), and an edge between two of its leaves ($y_0$ and $y_k$), which is isomorphic to $H_{k,2}^*$.

Assume there is a multi-colored clique of size $k$ in $G$. The paths connecting previous demands in the solution of the \textsc{SLSN} instance are the same as Case \ref{case:one}. The path between $r$ and $y_k$ is $r$ -- $z_{\{k-1,k\}}$ -- $z_{\{v_{k-1},v_k\}}$ -- $x_{v_k,k-1}$ -- $y_k$. All the edges in this path is already in the previous paths, so the cost stays the same. The length of this path is $2+1+2k^2-2+4=2k^2+5<4k^2$, which satisfies the length bound.

Assume there is a solution for the \textsc{SLSN} instance $(G',L,H_{k,2}^*)$ with total cost less than $\left(\frac{5}{4}-\frac{\varepsilon}{4}\right)g(H_{k,2}^*)$. The proof of existing a subgraph of $G$ with $\varepsilon\cdot\binom{k}{2}$ edges is the same as Case \ref{case:disjoint}.

\text{}\\
\textbf{Case \ref{case:matching}}: $H_{k,k}$

In order to get $H_{k,k}$ as our demand graph, we have to slightly change the construction in Case \ref{case:disjoint}. We still first make a weighted graph $G_{k,k}=(V_{k,k},E_{k,k})$ and then transform it to the unit-length graph $G'$. For the vertex set $V_{k,k}$, we add another layer of vertices $V_0=\{l_{i,j}'\mid i,j\in[k],i\ne j\}$ in to $V_k^*$ before the first layer $V_1$. For the edge set $E_{k,k}$, we include all the edges in $E_k^*$, but change the edges in $E_1$ to length $1$ and cost $1$. We also add another edge set $E_0=\{\{l_{i,j}',r\}\mid i,j\in[k],i\ne j\}$. Each edge in $E_0$ has length $1$ and cost $1$.

The demands are $\{l_{i,j}',l_{i,j}\}$ for each $i,j\in[k]$ where $i\ne j$, as well as $\{y_0,y_k\}$. This is a matching of size $k(k-1)+1$, which is isomorphic to $H_{k,k}$. We still set the length bound to be $L=4k^2$, and set $g(H_{k,k})=6k^6-6k^5+3k^4+\frac{k^2}{2}+\frac{k}{2}$.

If there is a multi-colored clique of size $k$ in $G$, the construction for the solution in $G'$ is similar to Case \ref{case:disjoint}. For each $i,j\in[k]$ where $i\ne j$, the paths between $l_{i,j}'$ and $l_{i,j}$ becomes $l_{i,j}'$ -- $r$ -- $z_{\{i,j\}}$ -- $z_{\{v_i,v_j\}}$ -- $x_{v_i,j}$ -- $x_{v_i,j}'$ -- $l_{i,j}$ (i.e., one more layer before the root $r$). It is easy to see
that the length bound and size bound are still satisfied.

Assume there is a solution for the \textsc{SLSN} instance $(G',L,H_{k,k})$ with total cost less than $\left(\frac{5}{4}-\frac{\varepsilon}{4}\right)g(H_{k,k})$. The proof of existing a subgraph of $G$ with $\varepsilon\cdot\binom{k}{2}$ edges is the same as Case \ref{case:disjoint}, except the path between $l_{i,j}'$ and $l_{i,j}$ has one more layer.

\subsubsection{Case \ref{case:bipartite}: $\mathcal{H}_{2,k}$}

For any $\varepsilon>0$, assume there is a demand graph $H\in \mathcal{H}_{2,k}$ and a \textsc{Multi-Colored D}$k$\textsc{S} instance $G=(V,E)$ with coloring function $c$, factor $\varepsilon$, and parameter $k$. We create a unit-length and polynomial-cost \textsc{SLSN} instance $G'$ as following.

We again use the length-weighted graph $G_{2,k}$ constructed in Section \ref{sec:bipartite}. We change the cost of edges in $E_{22}$ to $4k^4(k-1)$, the cost of edges in $E_{12}\cup E_{xl}$ to $8k^4$, while keeping the cost equal to the length for the rest of the edges.

$G'$ is again a graph that each edge $e\in E_{2,k}$ is replaced by a $length(e)$-hop path. Where the cost of edges is divided equally for each hop. The demands are the same as the demands in Section \ref{sec:bipartite}, and $L$ is still $7$. The construction still takes $|V||H|$. The function $g$ is slightly different, where $g(H)=16k^6-16k^5-10k^2+11k+7|H|-7\cdot\mathds{1}_{\{r_1,r_2\}\in H}$.

Using the same solution as in the proof of Lemma \ref{lem:bipartiteCost}, we can see that, If there is a multi-colored clique of size $k$, then the \textsc{SLSN} instance has a solution with cost
\begin{align*}
&7|H|-7k^2+9k-7\cdot\mathds{1}_{\{r_1,r_2\}\in H}+k\cdot(4k^4(k-1)-1)+k(k-1)\cdot(8k^4-1)+\binom{k}{2}\cdot(8k^4-4)\\
=&16k^6-16k^5-10k^2+11k+7|H|-7\cdot\mathds{1}_{\{r_1,r_2\}\in H},
\end{align*}
because the cost of $k$ edges in $E_{22}$ in this solution is changed from $1$ to $4k^4(k-1)$, the cost of $\binom{k}{2}$ edges in $E_{12}$ in this solution is changed from $1$ to $8k^4$, and the cost of $k(k-1)$ edges in $E_{xl}$ in this solution is changed from $4$ to $8k^4$.

The other direction for the correctness is the following lemma.

\begin{lemma}
Let $S$ be an optimal solution for the \textsc{SLSN} instance $(G',L)$ with demand graph $H\in\mathcal{H}_{2,k}$. If $S$ has cost less than $\left(\frac{5}{4}-\frac{\varepsilon}{4}\right)g(H)$, then there is a subgraph of $G$ with $\varepsilon\cdot\binom{k}{2}$ edges.
\end{lemma}
\begin{proof}
For each $i,j\in[k]$ where $i\ne j$, let $P_{1,i,j}$ be a (arbitrarily chosen) path in $S$ which connects $r_1$ and $l_{i,j}$, and $P_{2,i,j}$ be a (arbitrarily chosen) path in $S$ which connects $r_2$ and $l_{i,j}$. Let $\mathcal{P}_1=\{P_{1,i,j}\mid i,j\in[k],i\ne j\}$, and $\mathcal{P}_2=\{P_{2,i,j}\mid i,j\in[k],i\ne j\}$. We can see that each of these paths must have exactly one edge between each two levels. Where paths in $\mathcal{P}_1$ have form $r_1$ -- $z_{\{i,j\}}$ -- $z_{\{u,v\}}$ -- $x_{u,j}$ -- $l_{i,j}$ with $c(u)=i$, $c(v)=j$, and $\{u,v\}\in E$. The paths in $\mathcal{P}_2$ have form $r_2$ -- $y_i$ -- $y_{v}$ -- $x_{v,j}$ -- $l_{i,j}$ with $c(v)=i$.

For each color $i\in[k]$, we can see that, in order to connect $r_2$ with all $l_{i,j}$, there must be at least one edge $\{y_i,y_v\}\in S\cap E_{22}$ with $v\in C_i$. Let $v_i=\arg\max_{v\in C_i}\left|\{y_i,y_{v}\}\cap\bigcup_{j\in[k]\setminus\{i\}}P_{2,i,j}\right|$ be the vertex which is in the most number of paths in $\mathcal{P}_2$. We will prove that the induced subgraph with vertex set $\{v_1,\mathellipsis,v_k\}$ has at least $\varepsilon\cdot\binom{k}{2}$ edges.

For each $i\in[k]$, the edge $\{y_i,y_v\}\in S\cap E_{22}$ can only appear in at most $k-1$ different paths, which are $P_{2,i,j}$ where $j\in[k]\setminus\{i\}$. Because $\{y_i,y_{v_i}\}$ appears in most number of paths in $\mathcal{P}_2$, any $\{y_i,y_v\}$ other than $\{y_i,y_{v_i}\}$ can appear in at most $\frac{k-1}{2}$ different paths. Let $x=|S\cap E_{22}|$. Let $T=\{(i,j)\mid i,j\in[k],i\ne j,\{y_i,y_{v_i}\}\in P_{2,i,j}\}$. Then, $|T|\ge k(k-1)-(x-k)\cdot\frac{k-1}{2}=\frac{3}{2}k(k-1)-\frac{k-1}{2}x$.

We know that there is at least one different edge in $E_{xl}$ for each $P_{1,i,j}\in\mathcal{P}_1$ where $(i,j)\in T$, thus $\left|E_{xl}\cap\bigcup_{(i,j)\in T}P_{1,i,j}\right|\ge|T|$. There is also at least one different edge in $E_{xl}$ for each $P_{2,i,j}\in\mathcal{P}_2$, thus $\left|E_{xl}\cap\bigcup_{i,j\in[k],i\ne j}P_{2,i,j}\right|\ge k(k-1)$. Let $y=|S\cap E_{xl}|$, then
\begin{align*}
\left|E_{xl}\cap\bigcup_{(i,j)\in T}P_{1,i,j}\cap\bigcup_{i,j\in[k],i\ne j}P_{2,i,j}\right|
&\ge\left|E_{xl}\cap\bigcup_{(i,j)\in T}P_{1,i,j}\right|+\left|E_{xl}\cap\bigcup_{i,j\in[k],i\ne j}P_{2,i,j}\right|-|S\cap E_{xl}|\\
&\ge|T|+k(k-1)-y=\frac{5}{2}k(k-1)-\frac{k-1}{2}x-y.
\end{align*}
For each $(i,j)\in T$, by looking at the form of $\mathcal{P}_1$ and $\mathcal{P}_2$, we know that $E_{xl}\cap P_{1,i,j}$ can not intersect with $P_{2,i',j'}$ with any $(i',j')\ne(i,j)$. And if $E_{xl}\cap P_{1,i,j}$ do intersect with $P_{2,i',j'}$, the intersection must be exactly one edge $\{x_{v_i,j},l_{i,j}\}$. Therefore, let $T'=\{P_{1,i,j}\mid(i,j)\in T,\{x_{v_i,j},l_{i,j}\}\in E_{xl}\cap P_{1,i,j}\cap P_{2,i,j}\}$, we have $|T'|\ge\frac{5}{2}k(k-1)-\frac{k-1}{2}x-y$.

Let $z=|S\cap E_{12}|$ and $R=\{\{v_i,v_j\}\mid P_{1,i,j}\in T',P_{1,j,i}\in T',P_{1,i,j}\cap E_{12}=P_{1,j,i}\cap E_{12}\}$. Because each path in $T'$ has an edge in $S\cap E_{12}$, and any edge $(z_{c(u),c(v)},z_{\{u,v\}})\in S\cap E_{12}$ can appear in at most two paths $P_{1,c(u),c(v)}$ and $P_{1,c(v),c(u)}$ in $T'$, we know that $|R|\ge|T'|-z$.

For any $\{v_i,v_j\}\in R$, because $\{x_{v_i,j},l_{i,j}\}\in P_{1,i,j}$ and $\{x_{v_j,i},l_{j,i}\}\in P_{1,j,i}$, by looking at the form of $P_{1,i,j}$ and the form of $P_{1,j,i}$, we know that $P_{1,i,j}\cap E_{12}=P_{1,j,i}\cap E_{12}$ can only be the edge $\{z_{\{i,j\}},z_{\{v_i,v_j\}}\}$, which means $\{v_i,v_j\}\in E$. Therefore $R\subseteq E$. Thus the only thing left is to show that $|R|\ge\varepsilon\cdot\binom{k}{2}$.

Because $|H|\le(k(k-1)+2)(k(k-1)+1)$, we have
\[\frac{k-1}{2}x+y+z\le\frac{\left(\frac{5}{4}-\frac{\varepsilon}{4}\right)g(H)}{8k^4}\le\left(\frac{5}{2}-\frac{\varepsilon}{2}\right)k(k-1).\]
Thus we have
\[|R|\ge|T'|-z\ge\frac{5}{2}k(k-1)-\frac{k-1}{2}x-y-z\ge\frac{5}{2}k(k-1)-\left(\frac{5}{2}-\frac{\varepsilon}{2}\right)k(k-1)\ge\varepsilon\cdot\binom{k}{2}.\]

Therefore the induced subgraph with vertex set $\{v_1,\mathellipsis,v_k\}$ has at least $\varepsilon\cdot\binom{k}{2}$ edges.
\end{proof}

\subsubsection{Case \ref{case:all}: $\mathcal{H}_k$}

For any small constant $\varepsilon>0$, we now want to construct a \textsc{SLSN} instance for a demand graph $H\in\mathcal{H}_k$ from a \textsc{Multi-Colored D}$k$\textsc{S} instance $(G=(V,E),c)$, factor $\varepsilon$, and parameter $k$. By definition of $\mathcal H_k$, for some $t\in[5]$ there is a graph $H^{(t)}$ of Case $t$ which is an induced subgraph of $H$. We use Lemma \ref{lem:hk} to find out the graph $H^{(t)}$. Let $(G^{(t)},L)$ be the \textsc{SLSN} instance obtained from applying our reduction for case $t$ from the \textsc{Multi-Colored D}$k$\textsc{S} instance $(G=(V,E),c)$. We want to construct a instance $(G',c',L)$ with demand graph $H$, and makes sure that
\begin{itemize}
\item If the \textsc{SLSN} instance $(G^{(t)},c^{(t)},L)$ has a solution with cost $g^{(t)}(H^{(t)})$, then the \textsc{SLSN} instance $(G',c',L)$ has a solution with cost $g(H)$.
\item If the optimal cost of the \textsc{SLSN} instance $(G',c',L)$ is less than $\left(\frac{5}{4}-\varepsilon\right)g(H)$, then the optimal cost of the \textsc{SLSN} instance $(G^{(t)},c^{(t)},L)$ is less than $\left(\frac{5}{4}-\frac{\varepsilon}{4}\right)g^{(t)}(H^{(t)})$.
\end{itemize}
If there is such a construction, then
\begin{itemize}
\item If there is a multi-colored clique in $G$ with size $k$, then the \textsc{SLSN} instance $(G',c',L)$ has a solution with cost $g(H)$.
\item If the optimal cost of the \textsc{SLSN} instance $(G',c',L)$ is less than $\left(\frac{5}{4}-\varepsilon\right)g(H)$, then there is a induced subgraph of $G$ with $k$ vertices and at least $\varepsilon\cdot\binom{k}{2}$ edges.
\end{itemize}

Which is what we need for Theorem \ref{thm:approxReduct}.

The graph $G'$ is basically graph $G^{(t)}$ with some additional vertices and edges appeared in $H\setminus H^{(t)}$. We first increase the cost for all the edges in $G^{(t)}$ by multiplicative factor $\left\lceil\frac{L|H|}{\varepsilon}\right\rceil$. For each vertex $v$ in $H$ but not in $H^{(t)}$, we add a new vertex $v$ to $G'$. For each edge $\{u,v\}\in H\setminus H^{(t)}$, we add a $L$-hop path between $u$ and $v$ to $G'$, each new edge has cost $1$. We set $g(H)=\left\lceil\frac{L|H|}{\varepsilon}\right\rceil\cdot g^{(t)}(H^{(t)})+L\cdot(|H|-|H^{(t)}|)$. This is computable in $poly(|H|)$ time.

The construction still takes $poly(|V||H|)$ time, because the construction for the previous cases takes $poly(|V||H^{(t)}|)$ time and the construction for Case \ref{case:all} takes $poly(|G^{(t)}||H|)$ time. Here $|H^{(t)}|\le|H|$, and we know that $|G^{(t)}|$ is polynomial in $|V|$ and $|H^{(t)}|$.

If instance $(G^{(t)},L,H^{(t)})$ has a solution with cost $g^{(t)}(H^{(t)})$, let the optimal solution be $S^{(t)}$. For each $e=\{u,v\}\in H\setminus H^{(t)}$, let the new $L$-hop path between $u$ and $v$ in $G'$ be $P_e$. Then $S^{(t)}\cup\bigcup_{e\in H\setminus H^{(t)}}P_e$ is a solution to $G'$ with cost $\left\lceil\frac{L|H|}{\varepsilon}\right\rceil\cdot g^{(t)}(H^{(t)})+L\cdot(|H|-|H^{(t)}|)=g(H)$.

If instance instance $(G',L,H)$ has a solution with cost less than $\left(\frac{5}{4}-\varepsilon\right)g(H)$, let the optimal solution be $S$. Since for each $e=\{u,v\}\in H\setminus H^{(t)}$, the only path between $u$ and $v$ in $G'$ within the length bound is the new $L$-hop path $P_e$. Any valid solution must include all these $P_e$, which in total costs $L\cdot(|H|-|H^{(t)}|)$. In addition, for each demand $\{u,v\}$ which is also in $H^{(t)}$, any path between $u$ and $v$ in $G'$ within the length bound will not include any new edge, because otherwise it will contain a $L$-hop path, and have length more than $L$. Therefore, $S\setminus\bigcup_{e\in H\setminus H^{(t)}}P_e$ is a solution to $G^{(t)}$ with cost less than
\begin{align*}
&\frac{1}{\left\lceil\frac{L|H|}{\varepsilon}\right\rceil}\left(\left(\frac{5}{4}-\varepsilon\right)g(H)-L\cdot(|H|-|H^{(t)}|)\right)\\
=&\frac{1}{\left\lceil\frac{L|H|}{\varepsilon}\right\rceil}\left(\left(\frac{5}{4}-\varepsilon\right)\left(\left\lceil\frac{L|H|}{\varepsilon}\right\rceil\cdot g^{(t)}(H^{(t)})+L\cdot(|H|-|H^{(t)}|)\right)-L\cdot(|H|-|H^{(t)}|)\right)\\
=&\left(\frac{5}{4}-\varepsilon\right)\cdot g^{(t)}(H^{(t)})+\frac{L\cdot(|H|-|H^{(t)}|)\cdot\left(\frac{5}{4}-\varepsilon-1\right)}{\left\lceil\frac{L|H|}{\varepsilon}\right\rceil}\\
\le&\left(\frac{5}{4}-\varepsilon\right)\cdot g^{(t)}(H^{(t)})+\frac{L\cdot|H|\cdot\frac{1}{4}}{\frac{L|H|}{\varepsilon}}\\
\le&\left(\frac{5}{4}-\varepsilon\right)\cdot g^{(t)}(H^{(t)})+\frac{\varepsilon}{4}\\
\le&\left(\frac{5}{4}-\frac{\varepsilon}{4}\right)g^{(t)}(H^{(t)}).\\
\end{align*}

Therefore Theorem \ref{thm:approxReduct} is proved.

\subsection{Proof of Theorem \ref{thm:approxHard}:}\label{app:analysis}


If $\mathcal{C}$ is a recursively enumerable class, and $\mathcal{C}\nsubseteq\mathcal{C}_\lambda\cup \mathcal{C}^*$ for any constant $\lambda$, then for every $k\ge2$, let $H_k$ be the first graph in $\mathcal{C}$ where $H_k$ is not a star and it has at least $2k^{10}$ edges. The time for finding $H_k$ is $f(k)$ for some function $f$. From Lemma \ref{lem:hk} we know that $H_k\in\mathcal{H}_k$, so that we can use Theorem \ref{thm:approxReduct} to construct the $\textsc{SLSN}_\mathcal{C}$ instance with demand $H_k$.

The parameter $p=|H_k|$ of the instance is only related with $k$, and the construction time is \textsc{FPT} from Theorem \ref{thm:approxReduct}. Therefore this is a \textsc{FPT} reduction from the \textsc{Multi-Colored D}$k$\textsc{S} problem with parameter $k$ and factor $\varepsilon$ to the unit-length polynomial-cost $\textsc{SLSN}_\mathcal{C}$ problem with approximation factor $\left(\frac{5}{4}-\varepsilon\right)$. From Corollary \ref{col:dks}, the unit-length polynomial-cost $\textsc{SLSN}_\mathcal{C}$ problem has no $\left(\frac{5}{4}-\varepsilon\right)$-approximation algorithm in $f(p)\cdot poly(n)$ time for any function $f$, assuming Gap-ETH.
\qed

\ifprocs
\else
\bibliography{SLSN}

\begin{thebibliography}{10}

\bibitem{abboud2018reachability}
Amir Abboud and Greg Bodwin.
\newblock Reachability preservers: New extremal bounds and approximation
  algorithms.
\newblock In {\em Proceedings of the Twenty-Ninth Annual {ACM-SIAM} Symposium
  on Discrete Algorithms, {SODA} 2018}, pages 1865--1883, 2018.

\bibitem{agrawal1995trees}
Ajit Agrawal, Philip Klein, and R~Ravi.
\newblock When trees collide: An approximation algorithm for the generalized
  steiner problem on networks.
\newblock {\em SIAM Journal on Computing}, 24(3):440--456, 1995.

\bibitem{babay2017timely}
Amy Babay, Emily Wagner, Michael Dinitz, and Yair Amir.
\newblock Timely, reliable, and cost-effective internet transport service using
  dissemination graphs.
\newblock In {\em 37th {IEEE} International Conference on Distributed Computing
  Systems, {ICDCS} 2017}, pages 1--12, 2017.

\bibitem{bateni2011approximation}
MohammadHossein Bateni, MohammadTaghi Hajiaghayi, and D{\'a}niel Marx.
\newblock Approximation schemes for steiner forest on planar graphs and graphs
  of bounded treewidth.
\newblock {\em Journal of the ACM (JACM)}, 58(5):21, 2011.

\bibitem{berman1994improved}
Piotr Berman and Viswanathan Ramaiyer.
\newblock Improved approximations for the steiner tree problem.
\newblock {\em Journal of Algorithms}, 17(3):381--408, 1994.

\bibitem{byrka2013steiner}
Jaros{\l}aw Byrka, Fabrizio Grandoni, Thomas Rothvoss, and Laura Sanit{\`a}.
\newblock Steiner tree approximation via iterative randomized rounding.
\newblock {\em Journal of the ACM (JACM)}, 60(1):6, 2013.

\bibitem{chalermsook2017gap}
Parinya Chalermsook, Marek Cygan, Guy Kortsarz, Bundit Laekhanukit, Pasin
  Manurangsi, Danupon Nanongkai, and Luca Trevisan.
\newblock From gap-eth to fpt-inapproximability: Clique, dominating set, and
  more.
\newblock In {\em Foundations of Computer Science (FOCS), 2017 IEEE 58th Annual
  Symposium on}, pages 743--754. IEEE, 2017.

\bibitem{charikar1999approximation}
Moses Charikar, Chandra Chekuri, To-yat Cheung, Zuo Dai, Ashish Goel, Sudipto
  Guha, and Ming Li.
\newblock Approximation algorithms for directed steiner problems.
\newblock {\em Journal of Algorithms}, 33(1):73--91, 1999.

\bibitem{chekuri2011set}
Chandra Chekuri, Guy Even, Anupam Gupta, and Danny Segev.
\newblock Set connectivity problems in undirected graphs and the directed
  steiner network problem.
\newblock {\em ACM Transactions on Algorithms (TALG)}, 7(2):18, 2011.

\bibitem{chitnis2017parameterized}
Rajesh Chitnis, Andreas~Emil Feldmann, and Pasin Manurangsi.
\newblock Parameterized approximation algorithms for directed steiner network
  problems.
\newblock {\em arXiv preprint arXiv:1707.06499}, 2017.

\bibitem{chlamtac2017approximating}
Eden Chlamt{\'{a}}c, Michael Dinitz, Guy Kortsarz, and Bundit Laekhanukit.
\newblock Approximating spanners and directed steiner forest: Upper and lower
  bounds.
\newblock In {\em Proceedings of the Twenty-Eighth Annual {ACM-SIAM} Symposium
  on Discrete Algorithms, {SODA} 2017}, pages 534--553, 2017.

\bibitem{dreyfus1971steiner}
Stuart~E Dreyfus and Robert~A Wagner.
\newblock The steiner problem in graphs.
\newblock {\em Networks}, 1(3):195--207, 1971.

\bibitem{feldman2006directed}
Jon Feldman and Matthias Ruhl.
\newblock The directed steiner network problem is tractable for a constant
  number of terminals.
\newblock {\em SIAM Journal on Computing}, 36(2):543--561, 2006.

\bibitem{feldmann2017complexity}
Andreas~Emil Feldmann and D{\'{a}}niel Marx.
\newblock The complexity landscape of fixed-parameter directed steiner network
  problems.
\newblock In {\em 43rd International Colloquium on Automata, Languages, and
  Programming, {ICALP} 2016}, volume~55 of {\em LIPIcs}, pages 27:1--27:14.
  Schloss Dagstuhl - Leibniz-Zentrum fuer Informatik, 2016.

\bibitem{fellows2009parameterized}
Michael~R Fellows, Danny Hermelin, Frances Rosamond, and St{\'e}phane Vialette.
\newblock On the parameterized complexity of multiple-interval graph problems.
\newblock {\em Theoretical Computer Science}, 410(1):53--61, 2009.

\bibitem{goemans1995general}
Michel~X Goemans and David~P Williamson.
\newblock A general approximation technique for constrained forest problems.
\newblock {\em SIAM Journal on Computing}, 24(2):296--317, 1995.

\bibitem{guo2014shallow}
Longkun Guo, Kewen Liao, and Hong Shen.
\newblock On the shallow-light steiner tree problem.
\newblock In {\em Parallel and Distributed Computing, Applications and
  Technologies (PDCAT), 2014 15th International Conference on}, pages 56--60.
  IEEE, 2014.

\bibitem{hajiaghayi2009approximating}
Mohammad~Taghi Hajiaghayi, Guy Kortsarz, and Mohammad~R Salavatipour.
\newblock Approximating buy-at-bulk and shallow-light k-steiner trees.
\newblock {\em Algorithmica}, 53(1):89--103, 2009.

\bibitem{hassin1992approximation}
Refael Hassin.
\newblock Approximation schemes for the restricted shortest path problem.
\newblock {\em Mathematics of Operations research}, 17(1):36--42, 1992.

\bibitem{jain2001factor}
Kamal Jain.
\newblock A factor 2 approximation algorithm for the generalized steiner
  network problem.
\newblock {\em Combinatorica}, 21(1):39--60, 2001.

\bibitem{karp1972reducibility}
Richard~M Karp.
\newblock Reducibility among combinatorial problems.
\newblock In {\em Complexity of computer computations}, pages 85--103.
  Springer, 1972.

\bibitem{kortsarz1997approximating}
Guy Kortsarz and David Peleg.
\newblock Approximating shallow-light trees.
\newblock Technical report, Association for Computing Machinery, New York, NY
  (United States), 1997.

\bibitem{kou1981fast}
L~Kou, George Markowsky, and Leonard Berman.
\newblock A fast algorithm for steiner trees.
\newblock {\em Acta informatica}, 15(2):141--145, 1981.

\bibitem{lorenz2001simple}
Dean~H Lorenz and Danny Raz.
\newblock A simple efficient approximation scheme for the restricted shortest
  path problem.
\newblock {\em Operations Research Letters}, 28(5):213--219, 2001.

\bibitem{naor1997improved}
Joseph Naor and Baruch Schieber.
\newblock Improved approximations for shallow-light spanning trees.
\newblock In {\em Foundations of Computer Science, 1997. Proceedings., 38th
  Annual Symposium on}, pages 536--541. IEEE, 1997.

\bibitem{robins2000improved}
Gabriel Robins and Alexander Zelikovsky.
\newblock Improved steiner tree approximation in graphs.
\newblock In {\em SODA}, pages 770--779, 2000.

\bibitem{zelikovsky1996better}
Alexander Zelikovsky.
\newblock Better approximation bounds for the network and euclidean steiner
  tree problems.
\newblock {\em University of Virginia, Charlottesville, VA}, 1996.

\bibitem{zosin2002directed}
Leonid Zosin and Samir Khuller.
\newblock On directed steiner trees.
\newblock In {\em Proceedings of the thirteenth annual ACM-SIAM symposium on
  Discrete algorithms}, pages 59--63. Society for Industrial and Applied
  Mathematics, 2002.

\end{thebibliography}
\fi

\end{document}